\documentclass[11pt,a4paper]{article}

\usepackage{mathtools}
\usepackage{authblk} 
\usepackage{algpseudocode,algorithmicx,algorithm}
\usepackage{qcircuit}
\usepackage{mathrsfs}
\usepackage{latexsym,bm}
\usepackage{amsmath,amsfonts,amsmath,amssymb,amsthm}
\usepackage{extarrows}

\usepackage{graphicx,subfigure,epstopdf,float}
\usepackage{enumerate,cases,multirow}
\usepackage{makecell}
\usepackage{caption}
\usepackage{tikz}
\usepackage{pgfplots}

\usepackage{longtable,colortbl,arydshln,threeparttable}
\definecolor{mygray}{gray}{.9}
\usepackage[shortlabels]{enumitem}

\usepackage{indentfirst}
\usepackage[top=25mm,bottom=20mm,left=25mm,right=20mm]{geometry}
\usepackage{cite}

\usepackage{listings}
\usepackage{makeidx}        
\usepackage{booktabs}
\usepackage[bookmarksnumbered,colorlinks,citecolor=red,linkcolor=red,hyperindex,linktocpage=true]{hyperref}

\newcommand{\ket}[1]{| #1 \rangle} 
\newcommand{\bra}[1]{\langle #1 |} 

\newcommand{\bbC}{\mathbb{C}}

\newcommand{\Bf}{{\boldsymbol{f}}}
\newcommand{\Bg}{{\boldsymbol{g}}}

\newcommand{\Bj}{{\boldsymbol{j}}}

\newcommand{\Bn}{{\boldsymbol{n}}}

\newcommand{\Br}{{\boldsymbol{r}}}

\newcommand{\Bu}{{\boldsymbol{u}}}
\newcommand{\Bv}{{\boldsymbol{v}}}
\newcommand{\Bw}{{\boldsymbol{w}}}

\newcommand{\BB}{{\boldsymbol{B}}}

\newcommand{\BD}{{\boldsymbol{D}}}
\newcommand{\BE}{{\boldsymbol{E}}}
\newcommand{\BF}{{\boldsymbol{F}}}

\newcommand{\BI}{{\boldsymbol{I}}}
\newcommand{\BJ}{{\boldsymbol{J}}}

\newcommand{\BM}{{\boldsymbol{M}}}

\newcommand{\Cf}{\mathcal{F}}

\newcommand{\Ch}{\mathcal{H}}
\newcommand{\Ci}{\mathcal{I}}
\newcommand{\Cj}{\mathcal{J}}

\newcommand{\Cm}{\mathcal{M}}
\newcommand{\Cn}{\mathcal{N}}

\newcommand{\Cu}{\mathcal{U}}
\newcommand{\Cv}{\mathcal{V}}

\def \onebf{\textbf{1}}
\def \zerobf{\textbf{0}}

\def \CT{\text{CNOT}}
\def \RX{\text{RX}}

\def \CRZ{\text{CRZ}}
\def \CRX{\text{CRX}}
\def \CRY{\text{CRY}}

\def \szz{\sigma_{00}}
\def \soo{\sigma_{11}}
\def \szo{\sigma_{01}}
\def \soz{\sigma_{10}}
\def \d {\mathrm{d}}
\def \D {\mathrm{d}}
\newcommand{\bbR}{\mathbb{R}}
\newcommand{\deltabf}{\boldsymbol{\delta}}
\newtheorem{theorem}{Theorem}[section]
\newtheorem{lemma}{Lemma}[section]

\newtheorem{assumption}{\bf Assumption}[section]
\newtheorem{remark}{\bf Remark}[section]
\begin{document}

	\title{Schr\"odingerization based Quantum Circuits for  Maxwell's Equation with time-dependent source terms}
	\author[1,2]{Chuwen Ma   \footnote{Corresponding author.} \thanks{chuwenii@sjtu.edu.cn}}
 \author[1,2,3]{Shi Jin   \thanks{shijin-m@sjtu.edu.cn}}
	\author[2,3,4]{Nana Liu   \thanks{nana.liu@quantumlah.org}}
	
	\author[1,2]{Kezhen Wang  \thanks{carson-w@sjtu.edu.cn}  }
	\author[1,2,3]{Lei  Zhang    \thanks{lzhang2012@sjtu.edu.cn}}
	\affil[1]{School of Mathematical Sciences,   Shanghai Jiao Tong University, Shanghai 200240, China.}
	\affil[2]{Institute of Natural Sciences, Shanghai Jiao Tong University, Shanghai 200240, China.}
	\affil[3]{Ministry of Education, Key Laboratory in Scientific and Engineering Computing, Shanghai Jiao Tong University, Shanghai 200240, China.}
	\affil[4]{University of Michigan-Shanghai Jiao Tong University Joint Institute, Shanghai 200240, China.}  

	\maketitle
	\begin{abstract}
	The Schr\"odingerisation method combined with the autonomozation technique  in \cite{cjL23} converts general non-autonomous linear  differential equations with non-unitary dynamics into systems of autonomous Schr\"odinger-type equations, via the so-called warped phase transformation that maps the equation into two higher dimension. 
	Despite the success of Schr\"odingerisation techniques, they typically require 
	the black box of the sparse Hamiltonian simulation, suitable for continuous-variable based  analog quantum simulation. For  qubit-based general quantum computing one needs to design the quantum circuits for practical implementation.

	This paper explicitly constructs a quantum circuit for Maxwell's equations with perfect electric conductor (PEC) boundary conditions and time-dependent source terms, based on Schr\"odingerization and autonomozation, with corresponding computational complexity analysis.
	Through initial value smoothing and high-order approximation to the delta function, the increase in qubits from the extra dimensions only requires minor rise in computational complexity, almost $\log\log {1/\varepsilon}$ where $\varepsilon$ is the desired precision.
	Our analysis demonstrates that quantum algorithms constructed using Schr\"odingerisation exhibit polynomial acceleration in computational complexity compared to the classical Finite Difference Time Domain (FDTD) format.
		
	\end{abstract}
	\section{Introduction}
    Maxwell's equations are fundamental for understanding electromagnetic phenomena across a broad spectrum of scales, from subatomic particles to galactic structures. The profound impact of electromagnetic field theory on both science and technology is evident, with applications that permeate nearly every aspect of modern life, particularly in fields such as electrical engineering, optics, wireless optical communication, and remote sensing.
    In the context of modern wireless communication, there is an exponential surge in demand for efficient signal processing capabilities. Traditional algorithms are increasingly inadequate to meet these growing requirements.
    Despite remarkable progress in addressing extensive physical systems through supercomputers \cite{PMBSK19,ZDZB19}, obtaining solutions within a feasible computation time is still intractable.

    Quantum technology has achieved significant progresses, notably demonstrating quantum advantage over classical computers\cite{Divi95,Shor94,Stea98,NA07,SLF24,ZWDCPL20,AABB19}.
    Among many potential applications, one particularly promising area is the utilization of quantum computers  to solve the time-dependent Schr\"odinger equations, which follow unitary evolutions and hence the wave functions can be coherently represented on quantum computers. 
    A variety of efficient quantum algorithms have been developed in Hamiltonian simulations \cite{BCC15,BCC17,LC17,Low2019qubitization,CGJ19,Berry2019Dyson,AnLin2021TimeHamiltonian,
    	An2022timedependent,fang2023time}. To this aim, it is convenient to 
    	rewrite the source-free Maxwell formulation into a Hamiltonian system based on the Riemann-Silberstein vectors for quantum simulation \cite{costa19,suau21,vaha2020,vaha2020b,va20,va202,Bui2022}.

   When source terms or complex boundary conditions arise in electromagnetic field systems, the time evolution of the system often ceases to be unitary. Therefore, it is essential to transform these systems into unitary dynamical systems\cite{JLY22a,JLY22b, ALL,ACL23,cjL23}.   
   Among the unitarisation approaches, the Schr\"odingerisation method introduced in \cite{JLY22a,JLY22b} provides a simple and generic framework for quantum simulation of all linear partial differential equations (PDEs) and ordinary differential equations (ODEs). 
   Its essence lies in employing a warped phase transformation that elevates the equations to one higher dimension. In the Fourier space, this transformation reveals a system of Schr\"odinger-type equations. 
   The approach has been expanded to address a wide array of problems, including open quantum systems in bounded domains with non-unitary artificial boundary conditions \cite{JLLY23ABC}, problems entailing physical boundary or interface conditions \cite{JLLY22BDID}, Maxwell's equations \cite{JLM23}, the Fokker-Planck equation \cite{JLLY24}, ill-posed scenarios such as the backward heat equation \cite{JLM242}, linear dynamical systems with inhomogeneous terms \cite{JLM24},  iterative linear algebra solvers \cite{JL22}, etc.
   Despite the significant success of Schr\"odingerisation,  the explicit design of the corresponding quantum circuits for most problems remains to be addressed. 
   Recently, quantum circuits based on  Schr\"odingerisation have been designed for  heat equations, advection equations \cite{HZ24} and  heat equations with boundary conditions \cite{JLY24circuit}.

 This paper details the explicit construction  of  quantum circuits derived from the quantum algorithms through the process of Schr\"odingerisation, in conjunction with the aotonomousization  techniques in \cite{cjL23}. The main contributions of this paper are summarized as follows:
 \begin{itemize}
 	\item 
 We propose a quantum algorithm for Maxwell's equations under Perfect Electric Conductor (PEC) boundary conditions and time-varying source terms. This is accomplished by employing a stretching transformation to derive a homogeneous ordinary differential equation (ODE) system \cite{JLM24}, followed by the process of Schr\"odingerisation. Additionally, we apply the autonomousization methodology in \cite{cjL23} to construct a time-independent (autonomous) Hamiltonian system in one higher dimension, characterized by a Hamiltonian operator with spatially varying coefficients.
 \item 
 We construct the quantum circuits corresponding to the proposed quantum algorithms.
 \item 
 We apply smooth extension and high-order approximations to the delta function, demonstrating that the increase in dimensionality resulting from Schr\"odingerisation and the technique of turning a time-dependent non-autonomous system to a time-independent autonomous only increases mildly  the  computational complexity  on quantum computations, 
 almost $\log\log {1/\varepsilon}$ with $\varepsilon$ the desired precision. We further show that transforming a source-driven ODE system into a homogeneous one by introducing auxiliary variables via a stretching transformation does not degrade the success probability of obtaining the target state (see Remark~\ref{eq:probalility of uf}),  
 thus maintaining the computational efficiency of the quantum algorithm. 
 \end{itemize}

	The rest of the paper is organized as follows. 
	In section 2, we present the matrix representation of Maxwell's equations with physical boundary conditions and time-varying source terms.
	In section 3, we give a brief review of the Schr\"odingerisation approach for general inhomogeneous linear ODEs induced from the discretization of Maxwell's equations, and turn the non-autonomous system into an autonomous one.
    In section 4, we present the detailed implementation of quantum circuits for the time-independent Hamiltonian system.
    In section 5, we show the computational cost of the explicit quantum circuit and demonstrate that the complexity of the quantum algorithm exhibits polynomial acceleration compared to classical algorithms. 
    In section 6,  we conduct numerical experiments to validate the feasibility of the proposed algorithm, specifically focusing on the accuracy of the recovery through Schr\"odingerisation and high-order convergence rates.

	Throughout the paper, we restrict the simulation to a finite time interval $t\in [0,T]$, and we use a 0-based indexing, i.e. $j= \{0,1,\cdots,N-1\}$, or $j\in [N]$, and  $\ket{j}\in \bbC^N$, to denote a vector with the $j$-th component being $1$ and others $0$.
	We shall denote the identity matrix and null matrix by $I$ and $\textbf{0}$, respectively,
	and the dimensions of these matrices should be clear from the context. Otherwise,
	The notation $I_N$ stands for the $N$-dimensional identity matrix, and $\textbf{1}$ denotes $2$-dimension identity matrix.
	\section{The matrix representation of Maxwell's equations}
	In this section, we consider Maxwell's equations for a linear homogeneous medium with constant permittivity and permeability which are set to be  $1$, in the presence of sources of charge $\rho$ and currents $\BJ$ as
	\begin{subequations}\label{eq:maxwell Eq EB}
		\begin{align}
			\frac{\partial }{\partial t} \BE -\nabla \times \BB &= -\BJ,\quad
			\frac{\partial }{\partial t} \BB + \nabla \times \BE = \textbf{0},\\
			\nabla \cdot \BB &=0, \quad \quad \quad 
			\nabla \cdot \BE = \rho,
   \label{eq:maxwell Gauss law, maxwell Thomson}
		\end{align}
	\end{subequations}  
	in the three dimensional domain $\Omega=[0,1]^3$.
	From Equation~\eqref{eq:maxwell Eq EB}, the Maxwell-Gauss equation (or Gauss's law) and the Maxwell-Thomson equation~\eqref{eq:maxwell Gauss law, maxwell Thomson} are 
	actually consequences of the other equations and charge conservation equation
	\begin{equation}
		\frac{\partial \rho}{\partial t}  = \nabla \cdot \BJ.
	\end{equation} 
     To simplify the notation, let
	\begin{align*}
		\Cf &= \bigg(
		E_x \ket{0}
		+ E_y \ket{1}
		+ E_z \ket{2}
		+ r^a \ket{3}
		+ B_x \ket{4}
		+ B_y \ket{5}
		+B_z \ket{6}
		+r^b \ket{7}
		\bigg),\\
		\Cj &= \bigg(
		J_x \ket{0}+J_y \ket{1}+J_z\ket{2}
		-\rho\ket{7}
		\bigg).
	\end{align*}
	Write Equation~\eqref{eq:maxwell Eq EB} in vector form as
	\begin{align}    \label{eq:maxwell matrix 1} 
		&\frac{\partial \mathcal{F}}{\partial t} = \Cm \Cf -\Cj= \begin{bmatrix}
			\textbf{0} &\Cm\\
			\Cm^{\dagger} &\textbf{0}
		\end{bmatrix}  \Cf-\Cj. \quad  
			\Cm= \begin{bmatrix}
			0 &-\partial_z &\partial_y &-\partial_x\\
			\partial_z &0 &-\partial_x &-\partial_y\\
			-\partial_y &\partial_x &0 &-\partial_z\\
			\partial_x &\partial_y &\partial_z &0
		\end{bmatrix}. 
	\end{align}
	By comparing \eqref{eq:maxwell Eq EB} and \eqref{eq:maxwell matrix 1}, it is evident that the auxiliary variable $r^a\equiv r^b \equiv 0$. We consider the perfect electric conductor (PEC) boundary condition, which takes a specific form because a perfect conductor supports surface charges and currents, preventing fields from penetrating the body \cite{ACL18}, i.e., 
	\begin{equation}\label{eq:EB perfect conductor}
		\Bn \times \BE = \textbf{0},\qquad  \Bn\cdot \BB=0, \quad \text{on}\; \partial \Omega,
	\end{equation}
	where $\Bn$ is the unit normal to the boundary $\partial \Omega$.
	
 \section{Schr\"odingerization for Maxwell's equations}
 
	In this section, we provide a brief review of the Schr\"odingerization \cite{JLY22a,JLY22b} of Maxwell's equations using Yee's scheme \cite{JLM23}. We choose a uniform spatial mesh size $\triangle x = \triangle y = \triangle z = M^{-1}$ with $M$ an even positive integer given by $M = 2^m$.
 
	\subsection{Notations of finite difference operator}
    Before presenting the discretization, we introduce some notation. We define the shift operators as follows:
	\begin{equation*}
		S^{+}\ket{j} = \ket{j+1},\quad 
		S^{-}\ket{j} = \ket{j-1},\quad 
		1\leq j\leq M-1.
	\end{equation*}
	It is straightforward to verify that the matrices can be expressed as
	\begin{align}
		S^{-} &= \sum_{j=1}^{2^m-1} \ket{j-1}\bra{j} = \sum_{j=1}^m
		\onebf^{\otimes(m-j)}\otimes \sigma_{01}\otimes \sigma_{10}^{\otimes(j-1)}
		\triangleq \sum_{j=1}^m s_j^{-}, \label{eq:sn} \\
		S^{+}&= \sum_{j=1}^{2^m-1} \ket{j}\bra{j-1} = \sum_{j=1}^m
		\onebf^{\otimes(m-j)}\otimes \sigma_{10}\otimes \sigma_{01}^{\otimes(j-1)}
		\triangleq \sum_{j=1}^m s_j^{+}, 
		\label{eq:sp} \\
		I^{r}&=\sum_{j=1}^{2^{m}-1}\ket{j}\bra{j}= \onebf^{\otimes m} - \sigma_{00}^{\otimes m}.
		\label{eq:Ir}
	\end{align}
	Here $\sigma_{ij} := \ket{i}\bra{j}$, $i,j=0,1$ is a $2\times 2$ matrix.
	These two shift operators satisfy $S^{+} = (S^{-})^{\dagger}$.
	For a single qubit, the Pauli matrices are 
	\begin{equation*}
		X = \begin{bmatrix}
			0 &1 \\
			1 &0
		\end{bmatrix} = \szo+\soz, \quad 
		Y = \begin{bmatrix}
			0 &-i \\
			i &0 
		\end{bmatrix} = -i\szo+i\soz, \quad 
		Z = \begin{bmatrix}
			1 &0 \\
			0 &-1
		\end{bmatrix}=\szz-\soo.
	\end{equation*}
	
	\subsection{Discretization of Maxwell's equations}
	According to  Yee's lattice configuration \cite{Yee66,TafYee00}, the different components of the electromagnetic field
	and of the current densities are calculated at different cell center (half integer index) and cell vertices (integer index):
	\begin{align*}
		\BE_{\Bj}(t) &=(E_{x,\Bj},E_{y,\Bj},E_{z,\Bj})=
		(E_{x,j_1+\frac 12,j_2,j_3},\;
		E_{y,j_1,j_2+\frac 12,j_3},\;
		E_{z,j_1,j_2,j_3+\frac 12}),\\
		\BB_{\Bj}(t) &=(B_{x,\Bj},B_{y,\Bj},B_{z,\Bj})
		=(B_{x,j_1,j_2+\frac 12,j_3+\frac 12},\;
		B_{y,j_1+\frac 12,j_2,j_3+\frac 12},\;
		B_{z,j_1+\frac 12,j_2+\frac 12,j_3}).
	\end{align*}
	Correspondingly, the current densities, charge densities and the auxiliary variables are located according to  Yee's lattice configuration:
	\begin{equation*}
		\begin{aligned}
			\BJ_{\Bj}(t) &=(J_{x,\Bj},J_{y,\Bj},J_{z,\Bj})
			=(J_{x,j_1+\frac 12,j_2,j_3},\;
			J_{y,j_1,j_2+\frac 12,j_3},\;
			J_{z,j_1,j_2,j_3+\frac 12}),\\
			\rho_{\Bj}(t) &= \rho_{j_1,j_2,j_3},\quad 
			r_{\Bj}^a =r^a_{j_1,j_2,j_3},\quad 
			r_{\Bj}^b = r^b_{j_1+\frac 12,j_2+\frac 12,j_3+\frac 12}.
		\end{aligned}
	\end{equation*}
	Following Yee's algorithm, one gets the semi-discrete system:
	\begin{align}
		\frac{\D \BE_{h}}{\D t} 
		- \nabla_h \times \BB_{h} &= -\BJ_{h},   \quad \frac{ \D  r_h^a}{\D t} = \nabla_h \bm{\cdot} \BB_h, \label{eq:semi-Eh} \\
		\frac{\D \BB_{h}}{\D t}  + \nabla_h \times \BE_{h}  &= \textbf{0},\quad 
		\hspace{6mm} \frac{\D r_h^b}{\D t}  = \nabla_h \bm{\cdot} \BE_h +\rho_h,\label{eq:semi-Bh}
	\end{align}
	where $\BE_h$, $\BB_{h}$, $\BJ_h$, $r^a_h$ and $r^b_h$ are the collections of $\BE_{\Bj}$, $\BB_{\Bj}$, $\BJ_{\Bj}$, $r^a_{\Bj}$, and $r^b_{\Bj}$ for $0\leq j_1,j_2,j_3<M$,   
	\begin{equation}\label{eq:EhBhJhrh}
		\begin{aligned}
			\BE_h = \begin{bmatrix}
				\sum_{\Bj} E_{x,\Bj}\ket{\Bj}\\
				\sum_{\Bj} E_{y,\Bj}\ket{\Bj}\\
				\sum_{\Bj} E_{z,\Bj}\ket{\Bj}
			\end{bmatrix}, \quad
			\BB_h = \begin{bmatrix}
				\sum_{\Bj} B_{x,\Bj}\ket{\Bj}\\
				\sum_{\Bj} B_{y,\Bj}\ket{\Bj}\\
				\sum_{\Bj} B_{z,\Bj}\ket{\Bj}
			\end{bmatrix}, \quad
			\BJ_h = \begin{bmatrix}
				\sum_{\Bj} J_{x,\Bj}\ket{\Bj}\\
				\sum_{\Bj} J_{y,\Bj}\ket{\Bj}\\
				\sum_{\Bj} J_{z,\Bj}\ket{\Bj} 
			\end{bmatrix}, \quad 
			r_h^{\alpha} = \sum_{\Bj} r_{\Bj}^{\alpha}\ket{\Bj}, \;\alpha=a,b. 
		\end{aligned}
	\end{equation}
	The perfect electric conductor (PEC) boundary condition is given by
	\begin{equation*}
		\begin{aligned}
			&E_{x,j_1+\frac 12,j_2,j_3}\big|_{j_2=0\;\text{or}\;j_3=0} \equiv0,
			\quad 
			&&E_{y,j_1,j_2+\frac 12,j_3}\big|_{j_1=0\;\text{or}\;j_3=0} \equiv0,
			\quad 
			&&E_{z,j_1,j_2,j_3+\frac 12}\big|_{j_1=0\;
				\text{or}\;j_2=0}\equiv0,\\
			&B_{x,j_1,j_2+\frac 12,j_3+\frac 12}\big|_{j_1=0}\equiv0,
			\quad 
			&&B_{y,j_1+\frac 12,j_2,j_3+\frac 12}\big|_{j_2=0}\equiv0,
			\quad 
			&&B_{z,j_1+\frac 12,j_2+\frac 12,j_3}\big|_{j_3=0}\equiv0,
		\end{aligned}
	\end{equation*}
	for $0\leq j_1,\,j_2,\,j_3<M$ and $0\leq t\leq T$. To satisfy the boundary condition, we set
	\begin{equation*}
		\begin{aligned}
			J_{x,j_1+\frac 12,j_2,j_3}\big|_{j_2=0\;\text{or}\;j_3=0} &\equiv0,
			\quad 
			J_{y,j_1,j_2+\frac 12,j_3}\big|_{j_1=0\;\text{or}\;j_3=0} &&\equiv0,
			\quad 
			J_{z,j_1,j_2,j_3+\frac 12}\big|_{j_1=0\;
				\text{or}\;j_2=0}&&&\equiv0.
		\end{aligned}
	\end{equation*}  
	The discrete curl operator $\nabla_h \times$ is defined using central differences, which are also applied in the divergence operators. We define the difference matrices as follows
	\begin{align}
		\BD^{-}&=\sum_{i=0}^{M-2}\ket{i}\bra{i+1} -\sum_{i=1}^{M-1}\ket{i}\bra{i} = S^{-} - \BI^{r}, \label{eq:def FM}\\
		\BD^{+}&=\sum_{i=1}^{M-1}\ket{i}\bra{i} -\sum_{i=1}^{M-1} \ket{i}\bra{i-1} =\BI^{r}-S^{+}.\label{eq:def MF} 
	\end{align}
	It is evident that 	$\BD^{-} = -(\BD^{+})^{\dagger}$. Using \eqref{eq:def FM} and \eqref{eq:def MF}, we define the following matrices
	\begin{align}
		\BD_{x}^{+} = \frac{\BI_{M}\otimes \BI_{M}\otimes \BD^{+}}{\triangle x},\quad 
		\BD_{y}^{+} = \frac{\BI_{M}\otimes \BD^{+}\otimes\BI_{M}}{\triangle y},\quad 
		\BD_{z}^{+} = \frac{\BD^{+}\otimes \BI_{M}\otimes \BI_{M}}{\triangle z},\\
		\BD_{x}^{-} = \frac{\BI_{M}\otimes \BI_{M}\otimes \BD^{-}}{\triangle x},\quad 
		\BD_{y}^{-} = \frac{\BI_{M}\otimes \BD^{-}\otimes\BI_{M}}{\triangle y},\quad 
		\BD_{z}^{-} = \frac{\BD^{-}\otimes \BI_{M}\otimes \BI_{M}}{\triangle z}.
	\end{align}
	Next, we define the discrete curl operator
	\begin{equation*}
		\begin{aligned}
			M_{\nabla_h\times}^E&=
			\begin{bmatrix}
				\textbf{0}   &-\BD_z^{+}        &\BD_y^{+}        &-\BD_x^{+} \\
				\BD_z^{+}  		 &\textbf{0}    &-\BD_x^{+}       &-\BD_y^{+} \\
				-\BD_y^{+}   	 &\BD_x^{+}		&\textbf{0}   &-\BD_z^{+} \\
				\BD_x^{+}         &\BD_y^{+}        &\BD_z^{+}        &\textbf{0}
			\end{bmatrix},\quad 
			M_{\nabla_h \times}^B &=\begin{bmatrix}
				\zerobf       &\BD_z^{-} & -\BD_y^{-} & \BD_x^{-}\\
				-\BD_z^{-}  &\zerobf       &\BD_x^{-}  &\BD_y^{-}\\
				\BD_y^{-} & -\BD_x^{-} &\zerobf       &\BD_z^{-}\\
				-\BD_x^{-}  &-\BD_y^{-}  &-\BD_z^{-}  &\zerobf
			\end{bmatrix}.
		\end{aligned}
	\end{equation*}
	It follows that $M_{\nabla_h\times}^B =-(M_{\nabla_h\times}^E)^{\dagger}$. The matrix representation of \eqref{eq:semi-Eh}-\eqref{eq:semi-Bh} can be rewritten as an $n=8M^3$ dimensional ODE system:
	\begin{align}\label{eq:matrix FDTD}
		\frac{\D }{\D t} \Bu &= A \Bu +\Bf ,\quad
		A = 
		\begin{bmatrix}
			\textbf{0} & \BM_{\nabla_h \times}^E \\
			\BM_{\nabla_h \times}^B  &\textbf{0}
		\end{bmatrix},
	\end{align}
	where 
	$\Bu = \ket{0} \otimes\BE_h + \ket{1}\otimes r_h^a + \ket{2}\otimes\BB_h +\ket{3}\otimes r_h^b$ and 
    $\Bf = \ket{0}\otimes \Bf_1 + \ket{1}\otimes \Bf_2$
    with $\Bf_1 = \begin{bmatrix}
        -\BJ_h\\ \textbf{0}
    \end{bmatrix}$, $\Bf_2 = \sum_{j=0}^2 \ket{j}\otimes \textbf{0}+\ket{3}\otimes \rho_h$,  and $\textbf{0}\in \bbR^{M^3}$ is a zero vector.
 
	\subsection{A review of general framework of Schr\"odingerisation}
	It is time to consider the Schr\"odingerization of the linear system with source terms introduced in \cite{JLM24}, 
	\begin{equation} \label{eq:inhomo linear Eq}
		\frac{\D }{\D t} \begin{bmatrix}
			\Bu \\
			\Br
		\end{bmatrix}=
		\begin{bmatrix}
			A &\BF\\
			\textbf{0}^{\top} &\textbf{0}
		\end{bmatrix}\begin{bmatrix}
			\Bu \\
			\Br
		\end{bmatrix}
		= \tilde A \begin{bmatrix}
			\Bu \\
			\Br
		\end{bmatrix}
		,\qquad
		\begin{bmatrix}
			\Bu(0)\\ \Br(0)
		\end{bmatrix} = \begin{bmatrix}
			\Bu_0\\ \Br_0
		\end{bmatrix},
	\end{equation}
	where $\Br = [c_0,c_0,\cdots,c_0]^{\top}$, $c_0 =\max\{ \max_{t\in [0,T]}\|\Bf\|_{l^{\infty}},1\}$ is a constant, and  $\BF=\sigma_{00}\otimes \BF_1 + \sigma_{11}\otimes \BF_2 $ with $\BF_1 = \text{diag}\{\Bf_1/c_0\}$
	and $\BF_2 = \text{diag}\{\Bf_2/c_0\}$.
	Since any matrix can be decomposed into a Hermitian matric and an anti-Hermitian matric, Equation \eqref{eq:inhomo linear Eq} can be expressed as
	\begin{equation}\label{eq:ODE1}
		\frac{\D}{\D t} \Bu_f = (H_1+iH_2) \Bu_f, 
		\quad 
		H_1 = \frac{1}{2}
		(\tilde A+\tilde A^{\dagger}),\quad 
		H_2 = \frac{1}{2i}
		(\tilde A-\tilde A^{\dagger}),
	\end{equation}
	where $\Bu_f = \ket{0}\Bu + \ket{1}\Br$ 
	with $H_1$ and $H_2$ defined by 
	\begin{equation}
		H_1 = \frac 12 \begin{bmatrix}
			\textbf{0} & \textbf{0} &\BF_1 &\textbf{0} \\
			\textbf{0} & \textbf{0} & \textbf{0} &\BF_2 \\
			\BF_1 &\textbf{0}  & \textbf{0} & \textbf{0} \\
			\textbf{0} &\BF_2 & \textbf{0} & \textbf{0}
		\end{bmatrix},\quad 
		H_2 = \frac{1}{2i} \begin{bmatrix}
			\textbf{0} & 2M_{\nabla_h \times}^E &\BF_1 &\textbf{0} \\
			2M_{\nabla_h \times}^B  &\textbf{0} & \textbf{0} &\BF_2\\
			-\BF_1&\textbf{0}   &\textbf{0} & \textbf{0} \\
			\textbf{0} &-\BF_2 &\textbf{0} & \textbf{0}
		\end{bmatrix},
	\end{equation}
	both of which are Hermitian.  
	Thus, the sparsity of $H_1$ equals to the sparsity of $\tilde A$.
	Using the warped phase transformation $\Bw(t,p) = e^{-p}\Bu_f$ for $p>0$ and extending the initial data to $p<0$, 
	Equation \eqref{eq:ODE1} is converted to a system of linear convection equations\cite{JLM24}:
	\begin{equation}\label{eq:up}
		\begin{aligned}
			\frac{\D}{\D t} \Bw &= -H_1 \partial_p \Bw + iH_2\Bw, \\
			\Bw(0) &= g(p)\Bu_f(0),\quad 
		\end{aligned}
	\end{equation}
	with the smooth initial data $g\in C^{r}(\bbR)$ for any integer $r\geq 2$ \cite{JLM242}. 
	
	Next, we employ the spectral method to discretize the $p$-domain. We select a sufficiently large domain $p\in [-\pi L,\pi L]$, $L>0$ such that  
    \begin{equation}\label{eq:scope of L}
        e^{-\pi L + |\lambda_{\max}(H_1)|T} \approx 0,
    \end{equation}
    where $\lambda_{\max}(H_1)=\max_{t\in (0,T)}|\lambda(H_1(t)))|$.
    This ensures that the wave initially supported within the domain remains so throughout the computation. 
    We set the uniform mesh size  $\triangle p = 2\pi L/N_p$ where $N_p=2^{n_p}$ is a positive even integer and grid points denoted by $-\pi L=p_0<\cdots<p_{N_p}=\pi L$. Define the vector $\Bw_h$ as the collection of the function $\Bw$ at these grid points by
	\begin{equation}
		\Bw_h = \sum_{k\in [N_p]}	\sum_{j\in [n]} \Bw_j(t,p_k) \ket{j}\ket{k},
	\end{equation}
	where $\Bw_j$ is the $j$-th component of $\Bw$.
	The one-dimensional basis functions for the Fourier spectral method are typically chosen as
	\begin{equation} \label{eq:phi nu}
		\phi_l^p(x) = e^{i\nu_l^p (x+\pi L)},\quad \nu_l^p =  (l-N_p/2)/L,\quad 0\leq l\leq N_p-1. 
	\end{equation}
	Using \eqref{eq:phi nu}, we define 
	\begin{equation}\label{eq:Dp}
		\Phi^p = (\phi_{jl}^p)_{N_p\times N_p} = (\phi_l^p(p_j))_{N_p\times N_p},  \quad
		D_{p} = \text{diag}\{\nu_0^p,\cdots,\nu_{N_p-1}^p\}.
	\end{equation}
	Considering the Fourier spectral discretization on $p$, one obtains
	\begin{equation}
		\frac{\D}{\D t} \Bw_h = -i(H_1\otimes P_p) \Bw_h + i(H_2\otimes I_{N_p})\Bw_h \triangleq -iH_p \Bw_h.
	\end{equation} 
	Here $P_p$ is the matrix representation of the momentum operator $-i\partial_p$  and is 
	defined by $P_p = \Phi^p D_{p} (\Phi^p)^{\dagger}$.

	\subsubsection{Turning a non-autonomous system into an autonomous one}
 
	In this subsection, we apply the autonomousization technique from \cite{cjL23} to transform the time-dependent system into a time-independent one. The approach involves introducing a new variable $s$, which reformulates the problem as a new system defined in one higher dimension, featuring {\it time-independent} coefficients, as demonstrated in the following theorem.
	
	\begin{theorem}
		For the non-autonomous system in Equation \eqref{eq:schro of tilde w}, introduce the following initial-value problem of an autonomous PDE
			\begin{equation}\label{eq:v(s)}
			\frac{\D}{\D t} \Bv = -\frac{\partial}{\partial s}\Bv -i H_p(s) \Bv,\quad 
			\Bv(0) = \delta(s) \Bw_h(0),\quad s\in \bbR.
		\end{equation}
	One can  recover $\Bw_h(t)$ from $\Bv(t,s)$ using 
	\begin{equation} \label{eq:tilde whd}
	\Bw_h(t) = \int_{-\infty}^{\infty} \Bv(t,s)\;\d s.
	\end{equation}
	\end{theorem}
	For the discretization in the $s$ domain, we first truncate the infinite $s$-domain to $[-\pi S, \pi S]$, ensuring that $S$ is sufficiently large so that $e^{-\pi S +T} \approx 0$. This allows us to approximate the problem using periodic boundary conditions. Additionally, we set  $H_p(s)=0$ for $s\leq 0$.
	We define the collection of the function $\Bv$, denoted by $\Bv_h$, at the points $-\pi S = s_0<\cdots <s_{N_s} = \pi S$, where $s_i = -\pi S + (i-1)\triangle s$ for $i\in [N_s]$. Here, $N_s = 2^{n_s}$ represents the number of grid points in the $s$ domain, and $\triangle s = \frac{2\pi S}{N_s}$ is the mesh size.
	Analogous to equations \eqref{eq:phi nu}-\eqref{eq:Dp}, we define
	\begin{equation*}
		\Phi^s = (\phi_{jl}^s)_{N_s\times N_s} = (e^{i\nu_l^s(s_j+\pi S)})_{N_s\times N_s},\quad 
	D_s = \text{diag}\{\nu_0^s,\cdots, \nu_{N_s-1}^s\},\quad \nu_l^s = \frac{l-\frac{N_s}{2}}{S},
	\end{equation*}
	and then 
	\begin{equation*}
		P_s =\Phi^s D_s (\Phi^s)^{-1},\quad 	
       H_S = \sum_{l=0}^{N_s-1} \big(H_1(s_l)\otimes \ket{l}\bra{l}\otimes P_p - H_2(s_l)\otimes \ket{l}\bra{l}\otimes I_{N_p}\big).
	\end{equation*}
	The numerical scheme for the discretization of $\Bv$ is given by 
	\begin{equation}\label{eq:Hamiltonian v}
		\begin{aligned}
			\frac{\D}{\D t} \Bv_h = -i( I_{n}\otimes P_s \otimes I_{N_p} + H_S )\Bv_h, \quad \quad 
			\Bv_h(0)  = \Bu_f(0) \otimes \deltabf_h(0) \otimes \Bg_h,
		\end{aligned}
	\end{equation}
   where $\Bg_h = [\Bg(p_0)\;\cdots\; \Bg(p_{N_p-1})]^{\top}$ and $\deltabf_h=[\delta_{h}(s_0),\cdots,\delta_h(s_{N_s-1})]^{\top}$.
   The compactly supported discrete Dirac delta function is defined as  
	\begin{equation*}
		\delta_{h}(s) = \frac{1}{\triangle s} \beta\left(\frac{s}{\triangle s}\right),
	\end{equation*}
	 where $\beta(\xi)$ is the unscaled approximation function supported on the interval $[-q,q]$ for some constant $q$. 
	 We recover the approximation to $ \Bw_h$,  denoted by $ \Bw_h^D$  such that 
	 \begin{equation}\label{eq:whd}
	 	 \Bw_h^D = \triangle s\sum_{l\in [N_s]} (I_{n} \otimes \bra{l} \otimes I_{N_p})\cdot \Bv_h.
	 \end{equation}
	 
	Define
	\begin{equation*}
		H = I_{n} \otimes P_s\otimes I_{N_p} +
		\sum_{l\in [N_s]} \big( H_1(s_l)\otimes \ket{l}\bra{l}\otimes D_p - H_2(s_l)\otimes \ket{l}\bra{l}\otimes I_{N_p}\big).
	\end{equation*}
	By changing variables to $\tilde{\Bv}_h=[I_n \otimes I_{N_s} \otimes (\Phi^p)^{\dagger}]\Bv_h$,  one gets
	\begin{equation}\label{eq:schro of tilde w}
		\begin{aligned}
			\frac{\D}{\D t} \tilde{\Bv}_h &= -i H \tilde{\Bv}_h,\quad \quad 
				    \tilde \Bv_h(0) & = \Bu_f(0) \otimes \deltabf_h(0) \otimes \tilde \Bg_h,
		\end{aligned}
	\end{equation}
 where  $\tilde \Bg_h = (\Phi^p)^{\dagger} \Bg_h$.
	At this point, a quantum simulation can be performed for the Hamiltonian system described by \eqref{eq:schro of tilde w}. The inverse Quantum Fourier transform is then applied to yield $\Bv_h = [I_n  \otimes I_{N_s} \otimes \Phi^p] \tilde \Bv_h$.

	\subsubsection{Recovery from Schr\"odingerisation}
	Finally,  one needs to recover  
	$\Bu$ from the warped transformation by selecting a suitable domain in $p$. We adopt the recovery strategy outlined in \cite{JLM24}.
	\begin{theorem}\label{thm:recovery u}
		Assume the eigenvalues of $H_1$ satisfy $\lambda_1(H_1)\leq \lambda_2(H_1)\cdots\leq \lambda_n(H_1)$, the solution to Equation
		\eqref{eq:inhomo linear Eq}  can be recovered from Equation \eqref{eq:up} by
		\begin{equation}\label{eq: recover uh e}
			\Bu = e^{p^*} \Bw(p^*),\quad \text{for}\;\text{any}\; p^*> p^{\Diamond},
		\end{equation}
		where $p^{\Diamond}\geq \max\{\lambda_n(H_1) T,0\}$ with $T$ the stop time of the evolution; or  by using the integration,
		\begin{equation}\label{eq:recover integration}
			\Bu = e^{p^{*}}\int_{p^{*}}^{\infty} \Bw(q)\;dq, \quad p^*>p^{\Diamond}.
		\end{equation}
	\end{theorem}
	In the case of Equation \eqref{eq:matrix FDTD}, we have $|\lambda(H_1)| \leq \frac{1}{2}$, the numerical electromagnetic field is recovered by 
	\begin{equation}\label{eq:recovery uh}
		\begin{aligned}
			\Bu_h \ket{k} &= e^{p_k}\big(I_n\otimes 
			\ket{k}\bra{k}\big)  \Bw_h^D  \\
			& = e^{p_k}\triangle s\big(I_n\otimes \ket{k}\bra{k}\big) \sum_{l\in [N_s]}\big(I_{n} \otimes \bra{l} \times I_{N_p}\big)  \Bv_h\\
			& = e^{p_k} \triangle s \sum_{l\in [N_s]} \big(I_{n} \otimes \bra{l}\otimes I_{N_p}\big)M_k \Bv_h
			  \qquad p_k>T/2,
		\end{aligned}
	\end{equation}
	where $M_k = I_{n\times N_s} \otimes \ket{k}\bra{k} $ is the projection measurement operator. 
	Here we remark that we use the quantum computer to evaluate 
	\begin{equation*}
		\Bv_h =\big(I_n  \otimes I_{N_s} \otimes \Phi^p\big) U  
                \big(I_n \otimes I_{N_s} \otimes (\Phi^p)^{\dagger} \big)
                \big(\Bu_f(0) \otimes \deltabf_h(0)  \otimes \Bg_h\big)
                \triangleq \Cu \Bv_h(0),
	\end{equation*}
	where $U = \exp(-i HT)$ is unitary and $\Cu = (I_n \otimes I_{N_s} \otimes \Phi^p ) U  (I_n  \otimes I_{N_s} \otimes (\Phi^p)^{\dagger})$.
	A measurement corresponding to the projection 
	$M_k =I_{n}\otimes I_{N_s} \otimes \ket{k}\bra{k}$ is performed to select the $\ket{k}$ component of the state $\ket{\Bv_h}$. Following the measurement, the summation is executed on a classical computer. The complete circuit for implementing the quantum simulation of $\ket{\Bv_h}$ is illustrated in Figure~\ref{schr_circuit}, where QFT (IQFT) denotes the (inverse) quantum Fourier transform.
		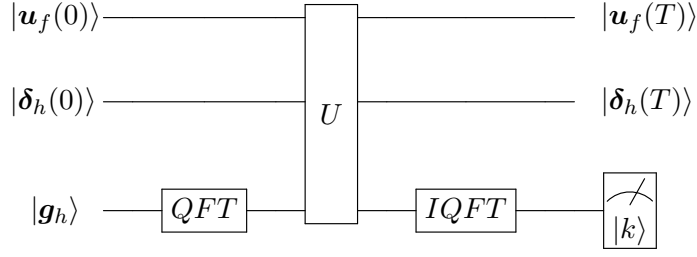
\begin{figure}[t!] 
		\centerline{
			\Qcircuit @C=1em @R=2em {
				\lstick{\hbox to 2.7em{$\ket{\Bu_f(0)}$\hss}}
				& \qw
				& \qw
				& \qw
				& \multigate{2}{U}
				& \qw
				& \qw
				& \qw
				& \qw  & \hbox to 2em{$\ket{\Bu_f(T)}$\hss} \\
				\lstick{\hbox to 2.7em{$\ket{\deltabf_h(0)}$\hss}}
				& \qw
				& \qw
				& \qw
				& \ghost{U}
				& \qw
				& \qw
				& \qw
				& \qw   & \hbox to 2em{$\ket{\deltabf_h(T)}$\hss} \\
				\lstick{\hbox to 2em{$\ket{\Bg_h}$\hss}} 
				& \qw
				& \gate{QFT}
				& \qw	
				& \ghost{U}	
				& \qw 
				& \gate{IQFT}
				& \qw	 
				& \qw  &\meterB{\ket{k}}
			}  
		}
		\caption{Quantum circuit for Schr\"odingerisation of Equation~\eqref{eq:Hamiltonian v}.}
		\label{schr_circuit}
	\end{figure}
	\subsection{A higher order improvement in the extended domains}
 In order to achieve $r$-th order convergence rates of spectral methods and $r$-th order approximation to the delta function, it is essential that $\beta(\xi)\in H^r(\bbR)$ is sufficiently smooth and 
 fulfills the condition 
	 \begin{equation}\label{eq:rth order delta}
	 	\big|\triangle s \sum_{l\in [N_s]} \delta_{h}(s_l - s)f(s_l) - f(s)\big| \leq C \triangle s^r
	 \end{equation}  
	 for $f(x)\in C^r(\bbR)$ \cite{TE04}.
  Additionally, the initial data in the extended domain $g(p)\in C^r(\bbR)$ should exhibit sufficient smoothness to ensure that $\Bw_h(p,t)\in C^r(\bbR)$.
Then,  we use the construction in \cite[Appendix]{LMC16} to obtain $r$-th order approximation to the delta function $\beta(\xi)\in H^r(\bbR)$. As an example,  we present the 3rd-order interpolating function $\beta(\xi)\in H^2(\bbR)$ from \cite{WMC21}, given by 
	\begin{equation}\label{eq:beta}
		\beta(\xi) = \begin{cases}
		  1-\frac{5}{2}|x|^2 +\frac{3}{2}|x|^3,\quad 0\leq |x| \leq 1,\\
		  \frac{1}{2}(2-|x|)^2(1-|x|), \quad 1\leq |x|\leq 2, \\
		  0, \quad \text{otherwise}.
		\end{cases}
	\end{equation}
    Meanwhile we choose $g(p)\in C^2(\bbR)$ in \cite{JLM242}:
	\begin{equation}\label{eq:gp} 
		g(p) =  \begin{cases}
            (\frac{19}{2}-\frac{19}{2}e^{-1})p^5+(\frac{49}{2}-23e^{-1})p^4+(\frac{35}{2}-\frac{29}{2}e^{-1})p^3+\frac{1}{2}p^2-p+1 \quad p\in(-1,0), \\
			e^{-|p|} \quad p\in (-\infty,-1]\cup[0,\infty).
		\end{cases}
	\end{equation}
    For general $r$, this will lead to essentially the spectral accuracy for the approximation in the extended space, if a spectral method is used \cite{JLM242}. 
 
	\section{Quantum circuit for  Maxwell's equations}
	The complete circuit for implementing the Schr\"odingerization method for Maxwell's equations is presented in Figure~\ref{schr_circuit}. In this section, we focus on the detailed construction of the quantum circuit for the unitary operation
	\begin{equation*}
		U = \exp(-i HT) = \exp(-i (I_{n}\otimes P_s\otimes I_{N_p} + H_S)T).
	\end{equation*}

    According to  Equation \eqref{eq:EhBhJhrh}, we define
	\begin{equation*}
		\BJ_{\alpha h} = \sum_{\Bj} \frac{J_{\alpha,\Bj}}{c_0}\ket{\Bj},\quad \alpha=x,y,z,\rho,
	\end{equation*}
	where $J_{\rho,\Bj} = \rho_{\Bj}$.
	For simplicity, we make the following assumptions.

	\begin{assumption}\label{ass:J}
	For the parameter $s_l$, $0\leq s_l<T$, $l\in \Ci_s$ with $|\Ci_s| \leq \mathscr{O}(n_s)$, the discretization source term $\BJ_{\alpha h}^l = \BJ_{\alpha h}(s_l)$ takes the value either $J_{\alpha ,1}^l$, or $J_{\alpha,0}^l$,  such that 
    \begin{equation*}
        \BJ_{\alpha h,\Bj}^l = 
        \begin{cases}
        J_{\alpha,1}^l\quad \text{for}\;\,j\in \Ci_{\alpha}^l,\\
         J_{\alpha,0}^l\quad \text{otherwise},
         \end{cases}\quad 
        \alpha = x,y,z,\rho.
    \end{equation*}
  In addition, $\BJ_{\alpha,h}^l \equiv 0$ for $l\in [N_s]\backslash \Ci_s$, and we have $|\Ci_{\alpha}| = \max_{l\in \Ci_s} |\Ci_{\alpha}^l| \leq  \mathscr{O}(m)$.
	\end{assumption}	
    
	To simplify notation, we denote 
	\begin{equation}\label{eq:notation}
		\Theta_{j} = \sigma_{j_{3m-1}j_{3m-1}} \otimes \cdots \otimes \sigma_{j_0 j_0}, \qquad 0\leq j\leq 2^{3m}-1,
	\end{equation}
	where $j = (j_{3m-1}\cdots j_0)=\sum_{i=0}^{3m-1} j_i 2^i$, $j_i\in \{0,1\}$ is the binary representation of $j$.
	By introducing 
	\begin{align}
		\BF^{\alpha}(s_l) &= \sum_{j \in \Ci_{\alpha}^l} \frac{1}{c_0}(J_{\alpha,1}^l - J_{\alpha,0}^l) \ket{j}\bra{j} + \frac{J_{\alpha,0}^l}{c_0} \onebf^{\otimes 3m}  \notag \\
		&= \sum_{j \in \Ci_{\alpha}^l} \frac{1}{c_0}(J_{\alpha,1}^l - J_{\alpha,0}^l)
		\Theta_j 
		+\frac{J_{\alpha,0}^l}{c_0} \onebf^{\otimes 3m}, \quad 
		\alpha = x,y,z,\rho, \label{eq:Falpha}
	\end{align}
	one has  the formula of $\BF$ as
	\begin{equation}\label{eq:Fsl}
		\BF(s_l) = \sigma_{3,x}\otimes \BF^x(s_l)
		+\sigma_{3,y}\otimes \BF^y(s_l)
		+\sigma_{3,z}\otimes \BF^z(s_l)
		+\sigma_{3,\rho} \otimes \BF^{\rho}(s_l),
	\end{equation}
	with $\sigma_{3,x} = \sigma_{00}^{\otimes 3}$,
	$\sigma_{3,y} = \sigma_{00}^{\otimes 2}\otimes 
	\sigma_{11}$, $\sigma_{3,z} = \sigma_{00}\otimes \sigma_{11}\otimes \sigma_{00}$, $\sigma_{3,\rho} = \sigma_{11}^{\otimes 3}$.

   We note that by employing the autonomousization technique described in \cite{cjL23}, we can transform the time evolution of the Hamiltonian system into a spatial variable-coefficient Schr\"odinger-type equation. Consequently, constructing the circuit diagram for the time evolution of the Hamiltonian system essentially involves creating a quantum circuit for variable-coefficient unitary operators. For more general cases of variable-coefficient matrices, we refer to \cite{SBHF24,STKY24}.

   The essential aspect of constructing the circuit lies in implementing the quantum circuit for the Hamiltonian operator $U(\tau) = \exp(iH\tau)$, where $\tau$ represents the time increment. This can be achieved by applying the first-order Lie-Trotter-Suzuki decomposition. Given that the Hamiltonian $H$ can be expressed as 
   \begin{align}
		H&= H_{Ds}+H_F+H_{\text{curl}}, \quad
        H_{Ds} = \onebf^{\otimes (3m+4)}\otimes P_s\otimes \onebf^{\otimes n_p},\label{eq:H}\\
		H_F&=  \frac{1}{2}  \sum_{l=0}^{N_s-1}  \big(X \otimes \BF(s_l) \otimes \ket{l}\bra{l}\otimes D_p 
           -Y \otimes \BF(s_l) \otimes \ket{l}\bra{l} \otimes  \onebf^{\otimes n_p} \big),\\
		H_{\text{curl}} & =  
		i(\sigma_{00}\otimes \sigma_{01}\otimes M_{\nabla_h\times}^E +
		\sigma_{00}\otimes\sigma_{10}\otimes M_{\nabla_h\times}^B)\otimes \onebf^{\otimes (n_p +n_s)}, \label{eq:Hcurl}
	\end{align}
    we introduce the operators 
    \begin{equation}
       \begin{aligned}
        \Cu_{1}(\tau) &= \exp(-iH_{D_s}\tau),\quad 
        \Cu_2(\tau) = \exp(-iH_F\tau),\quad 
        \Cu_3(\tau) &= \exp(-iH_{\text{curl}}\tau).\quad 
        \end{aligned}
    \end{equation} 
    Thus, the approximation for $U(\tau)$  can be formulated as
     \begin{equation}\label{eq:U1U2U3}
        \exp(-iH\tau) \approx \Cu_1(\tau)\Cu_2(\tau)\Cu_3(\tau).
    \end{equation}
	\subsection{Quantum circuit for $\Cu_1(\tau)$}
 
	From the definition of $P_s$, it is easy to construct the circuit of $\Cu_1$ as follows
	\begin{equation}\label{eq:u1_t}
		\begin{aligned}
			\Cu_1(\tau) &= \exp(-i \onebf^{\otimes 3m+4} \otimes \Phi^{s} D_s   (\Phi^s)^{\dagger} \otimes \onebf^{\otimes n_p})\\
			& = \onebf^{\otimes(3m+4)} \otimes \bigg(\Phi^s \exp(\frac{-i\tau }{2S}\sum_{l\in [N_s]}(l-\frac{N_s}{2})\ket{l}\bra{l}) (\Phi^s)^{\dagger}\bigg)\\
			& =\onebf^{\otimes(3m+4)} \otimes \bigg(\Phi^s  \big(\exp(\frac{i \tau N_s}{4S}) \onebf^{\otimes n_s}\big)
			\cdot
			\big(\exp(\frac{-i\tau}{2S}\sum_{l\in[N_s]} l\ket{l}\bra{l})\big) (\Phi^s)^{\dagger} \bigg) \otimes \onebf^{\otimes n_p}\\
			& \triangleq \onebf^{\otimes(3m+4)} \otimes \tilde \Cu_{1}(\tau) \otimes \onebf^{\otimes n_p}.
		\end{aligned}
	\end{equation}
	Using the binary representation of integers $l=(l_{n_s-1}\cdots l_0)=\sum_{j=0}^{n_s-1} l_j 2^j$, $l_j\in \{0,1\}$, one gets 
	\begin{equation}\label{eq:u1_2t}
		\begin{aligned}
			\exp(\frac{-i \tau}{2S} \sum_{l\in [N_s]} l \ket{l}\bra{l}) & = 
			\bigotimes_{j=0}^{n_s-1} \sum_{l_j =0,1} e^{-\frac{i\tau}{2S}2^j} \sigma_{l_jl_j} 
			= \bigotimes_{j=0}^{n_s-1} P(\frac{ -2^{j-1}}{S} \tau), 
		\end{aligned}
	\end{equation}
	where $P(\theta) = \begin{bmatrix}
		1 &0\\
		0 &e^{i\theta}
	\end{bmatrix}$ is the phase gate.
	According to Equation \eqref{eq:u1_t} and \eqref{eq:u1_2t}, one gets the circuit of $\tilde U_1(\tau)$ as shown in Figure \ref{tilde U1}.
	\begin{figure}
		\centerline{
			\Qcircuit @C=1em @R=2em {
			 \lstick{\hbox to 2.5em{$\ket{q_{n_s}}$\hss}}
			&  \multigate{3}{QFT} 
			& \qw
			& \qw
			& \qw
			& \gate{P(\frac{-2^{n_s-1}\tau}{S})}
			& \qw
			& \qw
			& \multigate{3}{IQFT} 
			& \qw \\
			\lstick{\hbox to 2.5em{$\cdots$\hss}}
			& \ghost{QFT}	
			& \qw
			& ...
			& 
			& ...
			&
			& \qw
			& \ghost{IQFT}	
			& \qw \\ 
			\lstick{\hbox to 2.5em{$\ket{q_2}$\hss}} 
			& \ghost{QFT}	
			& \qw
			& \qw	
			& \qw	
			& \gate{\;\,P(\frac{-2^{1}\tau}{S})\,\;\;}
			& \qw 
			& \qw
			& \ghost{IQFT}		 
			& \qw\\
		    \lstick{\hbox to 2.5em{$\ket{q_1}$\hss}} 
			& \ghost{QFT}	
			& \qw
			& \gate{I_{\frac{\tau N_s}{2S}}}
			& \qw	
			& \gate{\;\,P(\frac{-2^{0}\tau}{S})\,\;\;}
			& \qw 
			& \qw
			& \ghost{IQFT}		 
			& \qw
			}  
		}
		\caption{Quantum circuit for the operator $\tilde \Cu_1(\tau)$, where 
		$I_{\theta} = R_z(-\theta) P(\theta)$.}
		\label{tilde U1}
	\end{figure}
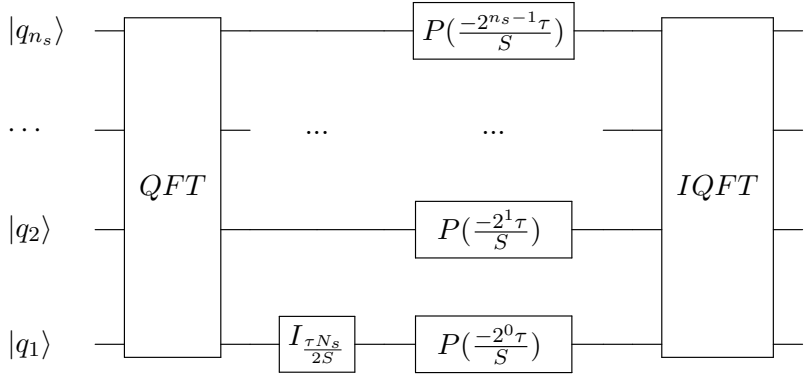
	
\subsection{Quantum circuits for $\Cu_2(\tau)$}
	  For any fixed $l=(l_{n_s} \cdots l_0)\in \Ci_s$ and $j=(j_{3m-1}\cdots j_0)\in \Ci_{\alpha}^l$, according to  the notation in Equation \eqref{eq:notation} and  the representation of spatial variable coefficient matrix in Equation \eqref{eq:Falpha}, we define 
     \begin{align}
        H_{X_p,l_j}^{\alpha} &= \frac{1}{2c_0} (J_{\alpha,1}^l-J_{\alpha,0}^l) X\otimes \sigma_{3,\alpha} \otimes
			\Theta_j  \otimes \Theta_l \otimes D_p
			, \\
   H_{X_p,l}^{\alpha} &=\frac{1}{2c_0}  J_{\alpha,0}^l X\otimes \sigma_{3,\alpha} \otimes  \onebf^{\otimes 3m}  \otimes \Theta_l \otimes D_p ,\\
   H_{Y,l_j}^{\alpha} & =\frac{1}{2c_0} (J_{\alpha,0}^l-J_{\alpha,1}^l) Y\otimes \sigma_{3,\alpha} \otimes
			\Theta_j \otimes \Theta_l \otimes \onebf^{\otimes n_p}
			,\\
   H_{Y,l}^{\alpha} &=\frac{1}{2c_0} (-J_{\alpha,0}^l) Y\otimes \sigma_{3,\alpha} \otimes  \onebf^{\otimes (3m)}  \otimes \Theta_l  \otimes \onebf^{\otimes n_p}.
    \end{align}
    Then one can rewrite $H_F$ as
    \begin{equation}\label{eq:H_F}
        H_F = \sum_{\alpha = x,y,z,\rho}\sum_{l\in \Ci_s} 
        \sum_{j\in \Ci_{\alpha}^l}
        \big(
        H_{X_p,l_j}^{\alpha} + H_{Y_p,l_j}^{\alpha} 
        \big)
        +H_{X_p,l}^{\alpha}+H_{Y_p,l}^{\alpha}.
    \end{equation}
    Using Trotter's splitting, one has 
    \begin{equation}\label{eq:V2}
        \Cu_2(\tau) \approx \prod_{\alpha =x,y,z,\rho} \prod_{\substack{l\in \Ci_s}} \prod_{j\in \Ci_{\alpha}^l}\Cu_{X_p,l_j}^{\alpha} \cdot \Cu_{Y_p,l_j}^{\alpha} \cdot\Cu_{X_p,l}^{\alpha} \cdot\Cu_{Y_p,l}^{\alpha}
        \triangleq \Cv_2(\tau),
    \end{equation}
    where
    \begin{align*}
        \Cu_{\beta,l_j}^{\alpha} =\exp(-i H_{\beta,l_j}^{\alpha} \tau),\quad 
        \Cu_{\beta,l}^{\alpha} = \exp(-i H_{\beta,l}^{\alpha} \tau),\quad \beta = \{X_p,Y_p\}.
    \end{align*}

    Define a multi-controlled RX($\theta$) gate acting on the ($m+1$)-th qubit. This gate becomes active when the qubits indexed in $\Ci^1$ are all set to 1 and the qubits indexed in $\Ci^0$ are all set to 0,
    denoted 
    \begin{equation}\label{eq:CRX}
    	\CRX_{m+1}^{\Ci^1,\Ci^0}(2\theta) = \exp\big(-i \theta X \otimes \tilde \sigma_{j_{m-1}j_{m-1}} \otimes \cdots \otimes \tilde\sigma_{j_0 j_0}\big),
    \end{equation}
    where $\tilde \sigma_{j_k}$, $k\in [m]$ is determined by 
    \begin{equation*}
    	\tilde \sigma_{j_kj_k} = \begin{cases}
    		\sigma_{00} \quad k\in\Ci^0,\\
    		\sigma_{11} \quad k\in \Ci^1,\\
    		\onebf\quad \;\;\,\text{otherwise}.
    	\end{cases}
    \end{equation*}
    Similarly, we can define $\CRY^{\Ci^1,\Ci^0}_{m+1}(2\theta)$ and $\CRZ^{\Ci^1,\Ci^0}_{m+1}(2\theta)$. 
    For the fixed $l=(l_{n_s} \cdots l_0)\in \Ci_s$ and $j=(j_{3m-1}\cdots j_0)\in \Ci_{\alpha}^l$, 
	the index  sets for control points are defined by 
     \begin{equation}\label{eq:index Ij Jl}
      \Ci_{j,3m}^{i} =\{k\in [3m]: j_k=i\},\quad 
      \Ci_{l,n_s}^{i} = \{k\in [n_s]:l_k =i\},\quad i=0,1.   
     \end{equation}

   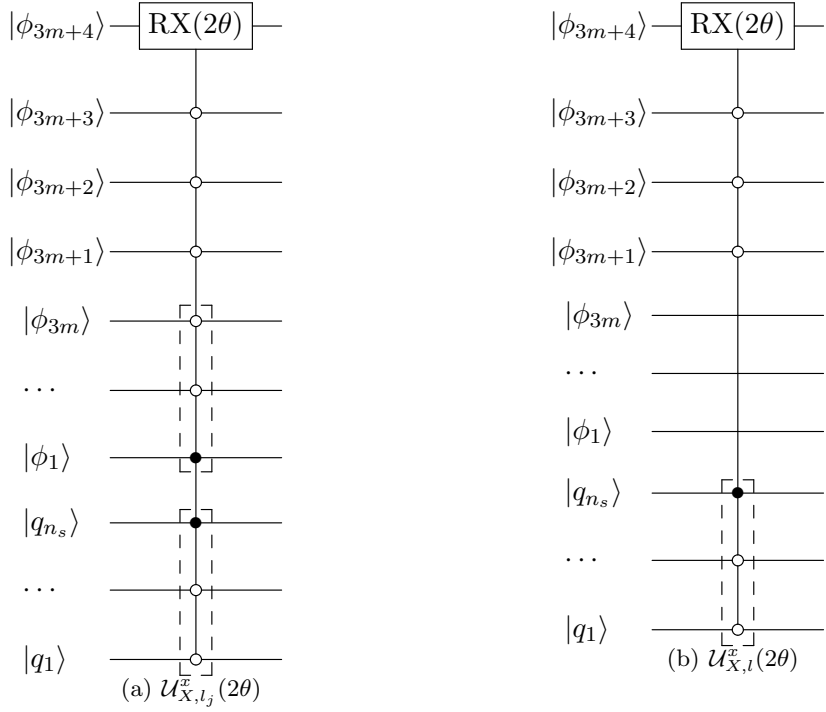
\begin{figure}[t!]
      \centering
	  \subfigure[$\Cu_{X,l_j}^{x}(2\theta)$]{
				\centering
				\Qcircuit @C=1em @R=2em {
					\lstick{\hbox to 3em{$\ket{\phi_{3m+4}}$\hss}} &  \gate{\RX(2\theta)}    & \qw    \\
					\lstick{\hbox to 3em{$\ket{\phi_{3m+3}}$\hss}} &  \ctrlo{-1}     & \qw      \\
					\lstick{\hbox to 3em{$\ket{\phi_{3m+2}}$\hss}} &  \ctrlo{-1}    & \qw      \\
					\lstick{\hbox to 3em{$\ket{\phi_{3m+1}}$\hss}} &  \ctrlo{-1}    & \qw     \\
					\lstick{\hbox to 2.5em{$\ket{\phi_{3m}}$\hss}} &  \ctrlo{-1}    & \qw     \\
					\lstick{\hbox to 2.5em{$\cdots$\hss}}          &  \ctrlo{-1}    & \qw     \\ 
					\lstick{\hbox to 2.5em{$\ket{\phi_{1}}$\hss}}  &  \ctrl{-1}   & \qw       \\
					\lstick{\hbox to 2.5em{$\ket{q_{n_s}}$\hss}}     & \ctrl{-1}	  & \qw      \\
					\lstick{\hbox to 2.5em{$\cdots$\hss}}      & \ctrlo{-1}	    & \qw      \\
					\lstick{\hbox to 2.5em{$\ket{q_1}$\hss}}         & \ctrlo{-1}	    &\qw    
                    \gategroup{7}{2}{5}{2}{.7em}{--}
                    \gategroup{10}{2}{8}{2}{.7em}{--}
				}   
			}
		\qquad \qquad \qquad \qquad \qquad \qquad
			\subfigure[$\Cu_{X,l}^{x}(2\theta)$]{
				\centering
				\Qcircuit @C=1em @R=2em {
					\lstick{\hbox to 3em{$\ket{\phi_{3m+4}}$\hss}} &  \gate{\RX(2\theta)}    & \qw    \\
					\lstick{\hbox to 3em{$\ket{\phi_{3m+3}}$\hss}} &  \ctrlo{-1}     & \qw      \\
					\lstick{\hbox to 3em{$\ket{\phi_{3m+2}}$\hss}} &  \ctrlo{-1}    & \qw      \\
					\lstick{\hbox to 3em{$\ket{\phi_{3m+1}}$\hss}} &  \ctrlo{-1}    & \qw     \\
					\lstick{\hbox to 2.5em{$\ket{\phi_{3m}}$\hss}} &  \qw    & \qw     \\
					\lstick{\hbox to 2.5em{$\cdots$\hss}}          &  \qw    & \qw     \\ 
					\lstick{\hbox to 2.5em{$\ket{\phi_{1}}$\hss}}  &  \qw   & \qw       \\
					\lstick{\hbox to 2.5em{$\ket{q_{n_s}}$\hss}}     & \ctrl{-4}	  & \qw      \\
					\lstick{\hbox to 2.5em{$\cdots$\hss}}      & \ctrlo{-1}	    & \qw      \\
					\lstick{\hbox to 2.5em{$\ket{q_1}$\hss}}         & \ctrlo{-1}	    &\qw    
                   \gategroup{10}{2}{8}{2}{.7em}{--}
				}   
			}
			 \caption{Quantum circuit for $\Cu_{X,l_j}^{x}(2\theta)$ and $\Cu_{X,l}^{x}(2\theta)$ with $\Ci_{j,3m}^{1} = \{0\}$, $\Ci_{l,n_s}^1 =\{n_s-1\}$.}
	 	\label{fig:tilde UX_lj}
    \end{figure} 
    Using the definition provided in \eqref{eq:CRX}, we define the operators corresponding to the circuits shown in Figure \ref{fig:tilde UX_lj} as follows
    \begin{align}
        \Cu_{X,l_j}^{\alpha}(2\theta) &= \exp(-i \theta X \otimes \sigma_{3,\alpha}\otimes \Theta_j \otimes \Theta_l)
        =\CRX_{3m+4+n_s}^{\Ci_{\alpha,1}^1,\Ci_{\alpha,1}^0}(2\theta) \quad \alpha = x,\,y,\,z,\,\rho, \label{eq:Uxlj}\\
         \Cu_{X,l}^{\alpha}(2\theta) &= \exp(-i \theta X \otimes \sigma_{3,\alpha}\otimes \onebf^{\otimes 3m} \otimes \Theta_l)
        =\CRX_{3m+4+n_s}^{\Ci_{\alpha,0}^1,\Ci_{\alpha,0}^0}(2\theta) \quad \alpha = x,\,y,\,z,\,\rho. \label{eq:Uxlj}
    \end{align}
    Here, the sets $\Ci_{\alpha,1}^k$ and $\Ci_{\alpha,0}^k$ are defined as
     \begin{align*}
    	\Ci_{\alpha,1}^k &= \{ n_s+3m+\Ci_{\alpha}^k\}\cup \{n_s+\Ci_{j,3m}^k\}\cup \Ci_{l,n_s}^k, \\
    	\Ci_{\alpha,0}^k &=  \{ n_s+3m+\Ci_{\alpha}^k\}\cup \Ci_{l,n_s}^k,
    \end{align*}
    for $k=0,\,1$. The sets $\Ci_{j,3m}^k$, $\Ci_{l,n_s}^k$ are defined in \eqref{eq:index Ij Jl}. Additionally, we have
    \begin{equation*}
    \begin{aligned}
    &\Ci_{x}^1=\emptyset, \quad 
     \Ci_y^1 = \{0\},\quad 
     \Ci_z^1 =\{1\},\quad 
     \Ci_{\rho}^1=\{0,1,2\},\quad 
     \Ci_{\alpha}^0 = \{0,1,2\}\backslash \Ci_{\alpha}^1,
     \quad \alpha = x,y,z,\rho,
     \end{aligned}
    \end{equation*}
    as determined by the definition of $\sigma_{3,\alpha}$.
     The circuit of $\Cu_{X,l_j}^{\alpha}(2\theta)$ and 
    $\Cu_{X,l}^{\alpha}(2\theta)$ are shown in Figure  \ref{fig:tilde UX_lj}.
    Similarly, one has $\Cu_{Y,l_j}^{\alpha}(2\theta)$ and $\Cu_{Y,l}^{\alpha}(2\theta)$.
    Thus, it yields
    \begin{align}
        \Cu_{Y_p,l_j}^{\alpha} = \Cu_{Y,l_j}^{\alpha}(\theta_l^{\alpha,1}\tau)\otimes \onebf^{\otimes n_p},\quad 
        \Cu_{Y_p,l}^{\alpha} = \Cu_{Y,l}^{\alpha}(\theta_l^{\alpha,0} \tau)\otimes \onebf^{\otimes n_p},\label{eq:UY_p}
    \end{align}
    where $\theta_l^{\alpha,1} = \frac{J_{\alpha,0}^l-J_{\alpha,1}^l}{c_0}$, $\theta_l^{\alpha,0}=-\frac{J_{\alpha,0}^l}{c_0}$.

	Noting that $D_p = \sum_{k=0}^{N_p-1} \frac{1}{L}(k-\frac{N_p}{2}) \ket{k}\bra{k}$,
    one gets the expression of $\Cu_{X_p,l_j}^{\alpha}$, $\Cu_{X_p,l}^{\alpha}$ as follows
    \begin{align*}
        \Cu_{X_p,l_j}^{\alpha} &=
        \bigg(\sum_{k\in [N_p]} \Cu_{X,l_j}^{\alpha}(k\tau \theta_{l_p}^{\alpha,1}) \otimes \ket{k}\bra{k} \bigg)
        \bigg(\Cu_{X,l_j}^{\alpha}(-N_p \tau \theta_{l_p}^{\alpha,1})\otimes \onebf^{\otimes n_p}\bigg)
        ,\\
        \Cu_{X_p,l}^{\alpha} &=
        \bigg(\sum_{k\in [N_p]} \Cu_{X,l}^{\alpha}(k\tau \theta_{l_p}^{\alpha,0}) \otimes \ket{k}\bra{k} \bigg)
        \bigg(\Cu_{X,l}^{\alpha}(-N_p \tau \theta_{l_p}^{\alpha,0})\otimes \onebf^{\otimes n_p}\bigg),
    \end{align*}
    where $\theta_{l_p}^{\alpha,1} = \frac{1}{L c_0}(J_{\alpha,0}^l-J_{\alpha,1}^l)$ and 
    $\theta_{l_p}^{\alpha,0} = \frac{J_{\alpha,0}^l}{2Lc_0}$.
    Using the binary representation of integers $k=(k_{n_p-1}\cdots k_0) = \sum_{g=0}^{n_p-1} 2^g k_g$, $k_g\in \{0,1\}$, one gets 
    $$\Cu_{X,l_j}^{\alpha}(k\theta) = (\Cu_{X,lj}^{\alpha}(\theta))^k =\prod_{g=1}^{n_p}(\Cu_{X,lj}^{\alpha}(\theta))^{k_g2^g}.$$
    Therefore, there holds
    \begin{equation}
        \begin{aligned}
        \sum_{k\in [N_p]} \Cu_{X,l_j}^{\alpha}(k\tau \theta_{l_p}^{\alpha,1}) \otimes \ket{k}\bra{k} 
        &= \sum_{k_{n_p}\cdots k_0} \prod_{g=0}^{n_p-1}  (\Cu_{X,l_j}^{\alpha}(\tau \theta_{l_p}^{\alpha,1}))^{k_g2^g}\otimes (\sigma_{k_{n_p}k_{n_p}}\otimes \cdots \otimes \sigma_{k_0k_0}) \\
        &=\coprod_{g=0}^{n_p-1} \left( \sum_{k_g=0,1}  (\Cu_{X,l_j}^{\alpha}( \tau \theta_{l_p}^{\alpha,1}))^{k_g2^g} \otimes \sigma_{k_gk_g}
        \right)\\
        & = \coprod_{g=0}^{n_p-1}\left(
        (\Cu_{X,l_j}^{\alpha}(\tau \theta_{l_p}^{\alpha,1}))^{k_g2^g} \otimes \ket{1}\bra{1}
        +\onebf^{\otimes 3m+4+n_s}\otimes \ket{0}\bra{0}
        \right).
        \end{aligned}
    \end{equation}
    The product $\coprod$ denotes the regular matrix product for the first register, which consists of $n_s+3m+4$ qubits, and the tensor product for the second register. Since $(\Cu_{X,l_j}^{\alpha}(\theta))^{2^g}$ can be implemented at a cost independent of $g$, the advantage of applying the binary representation of $k$ can be realized. The circuit for $\Cu_{X_p,l_j}^{\alpha}$ is illustrated in Figure \ref{tilde U2jl_gamma12}. The circuit for $\Cu_{X_p,l}^{\alpha}$ is obtained by replacing $\Cu_{X,l_j}^{\alpha}$ with $\Cu_{X,l}^{\alpha}$, as shown in Figure \ref{tilde U2l_gamma12}.
	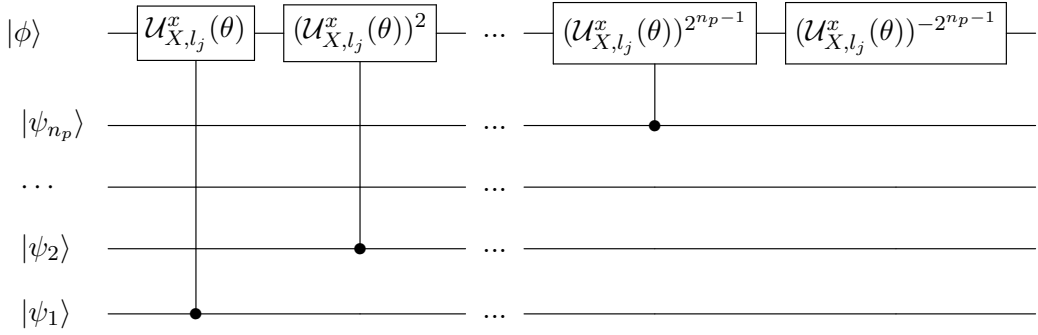
\begin{figure}[t!]
			\centerline{
			\Qcircuit @C=1em @R=2em {
				\lstick{\hbox to 3em{$\ket{\phi}$\hss}} &  \gate{\Cu_{X,l_j}^x(\theta)} & \gate{(\Cu_{X,l_j}^x (\theta))^{2}} & \qw   &...  &  & \gate{(\Cu_{X,l_j}^x(\theta))^{2^{n_p-1}}}     & \gate{(\Cu_{X,l_j}^{x}(\theta))^{-2^{n_p-1}}} & \qw    \\
				\lstick{\hbox to 2.5em{$\ket{\psi_{n_p}}$\hss}}  &  \qw	     & \qw        & \qw   &...  &  & \ctrl{-1}     & \qw  & \qw   \\
				\lstick{\hbox to 2.5em{$\cdots$\hss}}      &  \qw	     & \qw        & \qw   &...  &  & \qw           & \qw     & \qw\\
				\lstick{\hbox to 2.5em{$\ket{\psi_2}$\hss}}      & \qw 	     & \ctrl{-3}  & \qw   &...  &  & \qw           & \qw     & \qw \\
				\lstick{\hbox to 2.5em{$\ket{\psi_1}$\hss}}      & \ctrl{-4}	 & \qw	      & \qw   &...  &  & \qw           & \qw    & \qw  \\
			}   		
            }
			 	\caption{Quantum circuit for $\Cu_{X_p,l_j}^{x}$ with $\theta = \tau \theta_{l_p}^{\alpha,1}$.}
	 	\label{tilde U2jl_gamma12}
    \end{figure}
	\begin{figure}[htbp]
			\centerline{
			\Qcircuit @C=1em @R=2em {
				\lstick{\hbox to 3em{$\ket{\phi}$\hss}} &  \gate{\Cu_{X,l}^x(\theta)} & \gate{(\Cu_{X,l}^x (\theta))^{2}} & \qw   &...  &  & \gate{(\Cu_{X,l}^x(\theta))^{2^{n_p-1}}}     & \gate{(\Cu_{X,l}^{x}(\theta))^{-2^{n_p-1}}} & \qw    \\
				\lstick{\hbox to 2.5em{$\ket{\psi_{n_p}}$\hss}}  &  \qw	     & \qw        & \qw   &...  &  & \ctrl{-1}     & \qw  & \qw   \\
				\lstick{\hbox to 2.5em{$\cdots$\hss}}      &  \qw	     & \qw        & \qw   &...  &  & \qw           & \qw     & \qw\\
				\lstick{\hbox to 2.5em{$\ket{\psi_2}$\hss}}      & \qw 	     & \ctrl{-3}  & \qw   &...  &  & \qw           & \qw     & \qw \\
				\lstick{\hbox to 2.5em{$\ket{\psi_1}$\hss}}      & \ctrl{-4}	 & \qw	      & \qw   &...  &  & \qw           & \qw    & \qw  \\
			}   		
            }
			 	\caption{Quantum circuit for $\Cu_{X_p,l}^{x}$ with $\theta = \tau \theta_{l}^{\alpha,1}$.}
	 	\label{tilde U2l_gamma12}
    \end{figure}
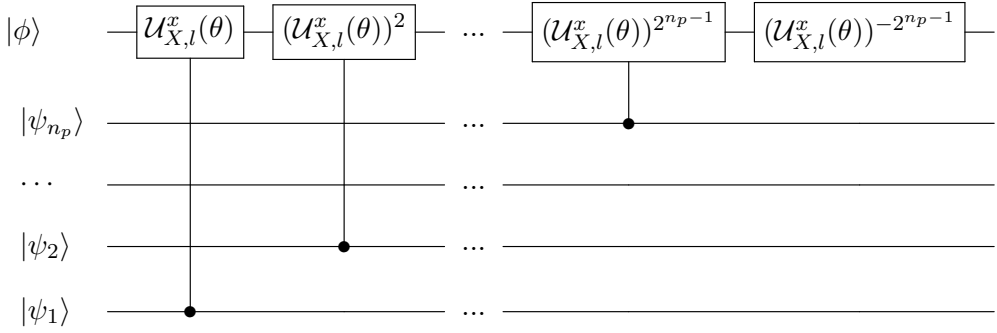

	\subsection{Quantum circuit for $\Cu_3(\tau)$}
		Before giving the specific circuit for $\Cu_3(\tau)$, we use the Bell basis to decompose a class of the matrix 
	shown in \cite{HZ24}.
	\begin{lemma}\label{lem:UZU}
		Given an operator of the form $S=e^{i\lambda}\sigma_{01}\otimes \sigma_{10}^{\otimes (n-1)}+e^{-i\lambda}\sigma_{10}\otimes \sigma_{01}^{\otimes (n-1)}$, where $\lambda$ is a real number, there exits a unitary matrix $U$ such that 
		\begin{equation*}
			S =	U_{n}^{\Ci}(-\lambda) \Lambda (U_n^{\Ci}(\lambda))^{\dagger},\quad  
		\end{equation*}
		where $\Lambda = Z \otimes \sigma_{11}^{\otimes(n-1)}$. The unitary matrix $U_n^{\Ci}(-\lambda)$ is defined as
		\begin{equation*}
			U_n^{\Ci}(-\lambda) =(\prod_{k\in \Ci} \CT_k^n)P_n(\lambda)H_n,
		\end{equation*}
		where $\Ci=\{1,2,\cdots,n-1\}$ is the index set,
		$H_n$ is the Hadamard gate acting on the $n$-th qubit, $P_n(\lambda)$ is the phase gate $P(\lambda)$ acting on the $n$-th qubit 
		and $CNOT_{k}^n$ is the CNOT gate acting on the $k$-th qubit controlled by the $n$-th qubit. 
	\end{lemma}
	To present the explicit quantum circuit for  $\Cu_3$, we first decompose $H_{\text{curl}}$ into
	\begin{equation}\label{eq:Hcurl_HxHyHz}
		\begin{aligned}
			H_{\text{curl}} &=i(\sigma_{00}\otimes \sigma_{01}\otimes M_{\nabla_h\times}^E +
			\sigma_{00}\otimes\sigma_{10}\otimes M_{\nabla_h\times}^B)\otimes \onebf^{\otimes (n_p+n_s)} \\
			&= (H_x+H_y+H_z)\otimes \onebf^{\otimes(n_s+n_p)},
		\end{aligned}
	\end{equation}
	where the matrices $H_x$, $H_y$ and $H_z$ are defined as follows
	\begin{align*}
		H_x &= \big(\sigma_{00}\otimes \sigma_{01}\otimes Y\otimes X \otimes \BD_x^{+} 
		-\sigma_{00}\otimes \sigma_{10}\otimes Y \otimes X \otimes \BD_x^{-}
		\big), \\
		H_y & = -\big(\sigma_{00}\otimes \sigma_{01}\otimes Y\otimes Z \otimes \BD_y^{+}
		-\sigma_{00}\otimes \sigma_{10}\otimes Y \otimes Z \otimes \BD_y^{-}
		\big)  , \\
		H_z & =
		\big(\sigma_{00}\otimes \sigma_{01} \otimes \onebf \otimes Y \otimes \BD_z^{+} 
		-\sigma_{00}\otimes \sigma_{10}\otimes\onebf \otimes Y \otimes \BD_z^{-}
		\big)  . 
	\end{align*}
	The circuit for $\Cu_3$ is approximated by $\Cu_3\approx (\Cu_x\Cu_y\Cu_z)\otimes \onebf^{\otimes (n_s+n_p)}$ with
	\begin{equation*}
		\Cu_x = \exp(-i H_x \tau), \quad \Cu_y = \exp(-i H_y \tau),\quad \Cu_z = \exp(-i H_z \tau).
	\end{equation*}
	From Equations \eqref{eq:sn} and \eqref{eq:sp}, we find that $H_x=H_{x_1} + H_{x_2}$, where $H_{x_1}$ and $H_{x_2}$ can be defined by
	\begin{align*}
		H_{x_1}& = 
		-\big(\frac{1}{\triangle x}\sigma_{00}\otimes \sigma_{01}\otimes Y\otimes X 
		\otimes \sum_{j=1}^m \onebf^{\otimes (3m-j)} 
		\otimes \sigma_{10}\otimes \sigma_{01}^{\otimes (j-1)} \notag \\
		&\hspace{5mm}+
		\frac{1}{\triangle x}\sigma_{00}\otimes \sigma_{10}\otimes Y\otimes X 
		\otimes \sum_{j=1}^m \onebf^{\otimes (3m-j)} 
		\otimes \sigma_{01}\otimes \sigma_{10}^{\otimes (j-1)}\big), \\
		H_{x_2} &=  
		\frac{1}{\triangle x}\sigma_{00}\otimes X \otimes Y\otimes X \otimes \onebf^{\otimes 2m}\otimes  \BI^r.
	\end{align*}
	We can express the tensor product of the Pauli matrices by
	\begin{align}
		Z\otimes Y\otimes X 
		&=  (\onebf\otimes U_Y\otimes H)\cdot Z_{\sigma,3}\cdot (\onebf\otimes U_Y^{\dagger}\otimes H), \label{eq:ZYZ}\\
		X\otimes Y \otimes X
		&=(H\otimes U_Y\otimes H)\cdot Z_{\sigma,3}\cdot  (H\otimes U_Y^{\dagger}\otimes H). \label{eq:XYX}
	\end{align}
	Here $U_Y$ is defined as $H P(\frac{\pi}{2})H$ and $Z_{\sigma,3}$ represents $Z^{\otimes 3}$.
	Applying Lemma~\ref{lem:UZU} yields
	\begin{equation*}
		\begin{aligned}
			H_{x_1} 
			&=  \tilde U_{ZYX}\cdot \sum_{j=1}^m\left( U_{3m+3}^{\Ci_{x,j}}(0)X_{3m+3,j}(\szz\otimes\frac{-Z_{\sigma,3}}{\triangle x}\otimes \onebf^{\otimes (3m-j)}\otimes \soo^{\otimes j})
			X_{3m+3,j}(U_{3m+3}^{\Ci_{x,j}}(0))^{\dagger}\right) \tilde U_{ZYX}^{\dagger}.
		\end{aligned}
	\end{equation*}
	In this expression, $\tilde U_{ZYX}$ is defined as $\onebf^{\otimes 2}\otimes U_Y\otimes H\otimes \onebf^{\otimes 3m}$. The term $U_{3m+3}^{\Ci_{x,j}}(\lambda)$ is given by 
	$(\prod\limits_{k\in \Ci_{x,j}}\CT_{k}^{3m+3}) P_{3m+3}(\lambda) H_{3m+3}$.
	The set $\Ci_{x,j}$ represents $\{1,2,\cdots,j\}$, 
	and $X_{3m+3,j}$ denotes the X gate acting on the $(3m+3)$-th and $j$-th qubit.
	Using the first-order Lie-Trotter-Suzuki decomposition, it yields
 \begin{equation*}
		\begin{aligned}
      \Cu_{x_1} & =
			\exp(-iH_{x_1}\tau)\\
   & \approx \tilde U_{ZYX}\cdot \prod_{j=1}^m 
   \left(U_{3m+3}^{\Ci_x,j}(0) X_{3m+3,j}
   \CRZ_{{3m+1\to3m+3}}^{\Ci_{x,j}^1,\Ci_{x,j}^0}(\frac{-2\tau}{\triangle x})  X_{3m+3,j} (U_{3m+3}^{\Ci_x,j}(0))^{\dagger}\right) \cdot 
			\tilde U_{ZYX}^{\dagger}\\
			& = 
			\tilde U_{ZYX} \prod_{j=1}^m W_j \tilde U_{ZYX}^{\dagger}
			 = \Cv_{x_1}.
		\end{aligned}
	\end{equation*}
   Here,  $\CRZ_{{3m+1\to 3m+3}}^{\Ci_{x,j}^1,\Ci_{x,j}^0}(2\theta)$ represents 
   a multi-gate $\text{RZ}_3(2\theta) = \exp(-i \theta Z_{\sigma,3})$, acting on the $(3m+1)$ to $(3m+3)$-th qubits controlled by $\Ci_{x,j}^1$ and $\Ci_{x,j}^0$ given by 
   \begin{align*}
       \Ci_{x,j}^1=\{1,2,\cdots,j\},\quad \Ci_{x,j}^0 =\{3m+4\}.
   \end{align*}
	The circuits of $W_j$ and $\Cv_{x_1}$ are shown in Figure~\ref{fig:cirUx1}.
	\begin{figure}[t]
		\centering
		\subfigure[$W_j$]{
			\includegraphics[width=0.48\linewidth]{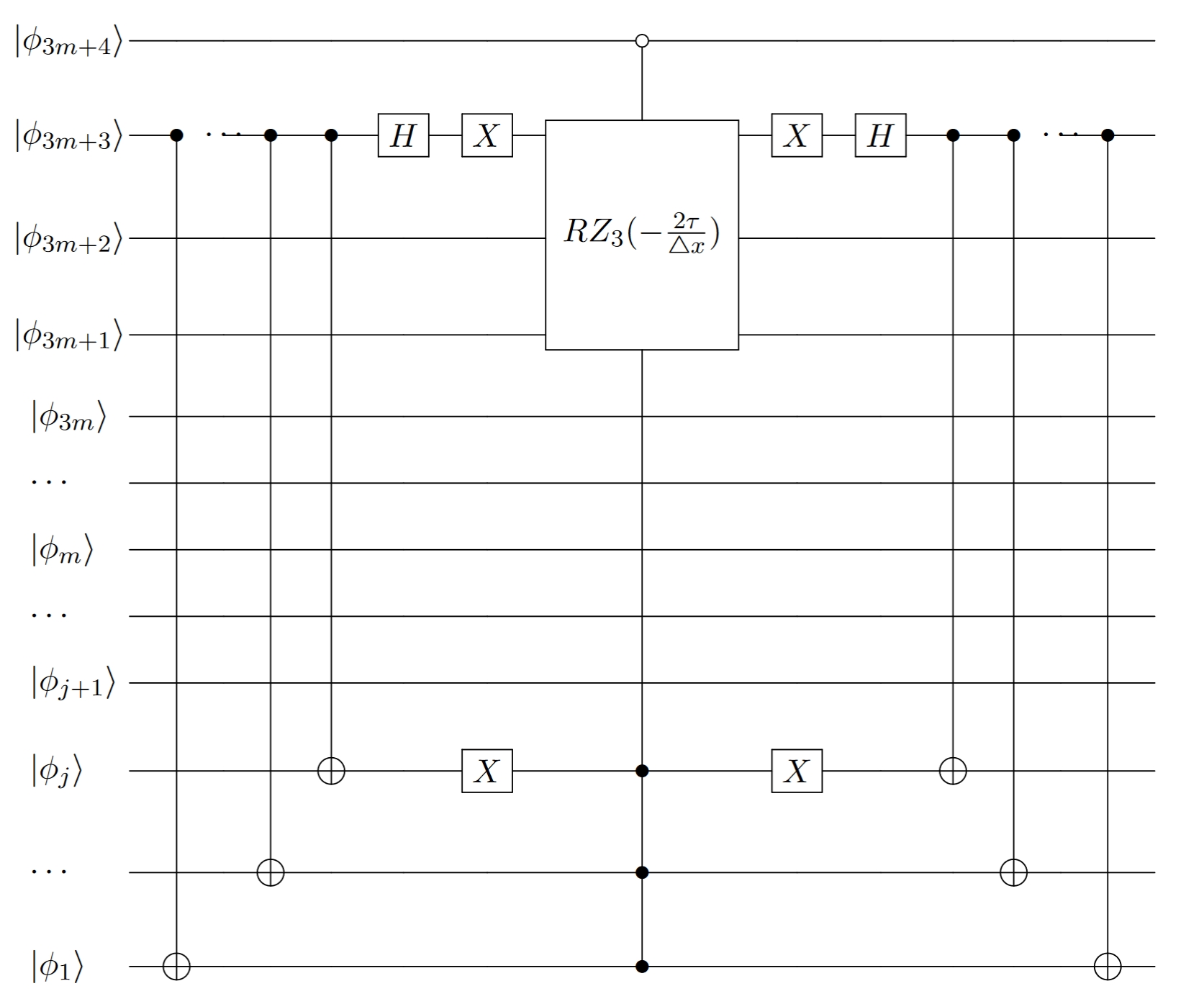}\label{fig:VI}}
		\subfigure[$\Cv_{x_1}$]{
			\includegraphics[width=0.49\linewidth]{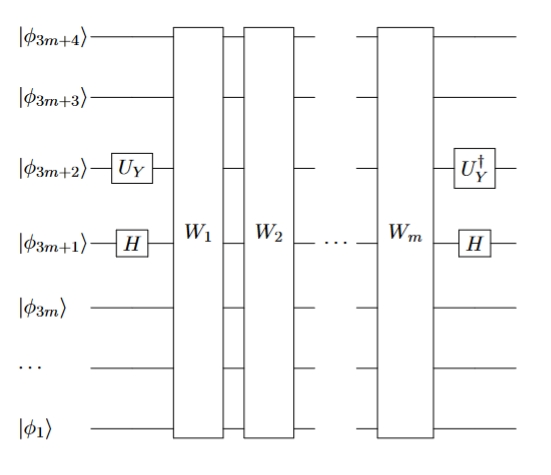}\label{fig:VII}}
		\caption{Quantum circuit for $W_j$ and $\Cv_{x_1}$.}
		\label{fig:cirUx1}
	\end{figure}
Noting that $\BI^r = \onebf^{\otimes m} - \sigma_{00}^{\otimes m}$ and Equation \eqref{eq:XYX}, one has  
	\begin{equation*}
      \begin{aligned}
      \Cu_{x_2} &= \exp(-i H_{x_2}\tau)\\
        &=  \tilde U_{XYX} \cdot\CRZ_{\sigma_{3m+1\to 3m+3}}^{\emptyset,\{3m+4\}}(\frac{-2\tau}{\triangle x})\cdot \CRZ_{\sigma_{3m+1\to 3m+3}}^{\emptyset,\{1,2,\cdots,m,3m+4\}}(\frac{2\tau}{\triangle x}) \cdot \tilde U_{XYX}^{\dagger},
        \end{aligned}
	\end{equation*}
     where $\tilde U_{XYX} = \onebf \otimes H\otimes U_Y \otimes H \otimes \onebf^{\otimes 3m}$.
	The circuit of $\Cu_{x_2}$ is shown in Figure \ref{fig:Ux2}, and the approximation of $\Cu_x$ is given by
 \begin{equation}\label{eq:Vx}
		\Cu_x =\Cu_{x_1}\Cu_{x_2}\approx \Cv_x = \Cv_{x_1}\cdot \Cu_{x_2}.
	\end{equation}
	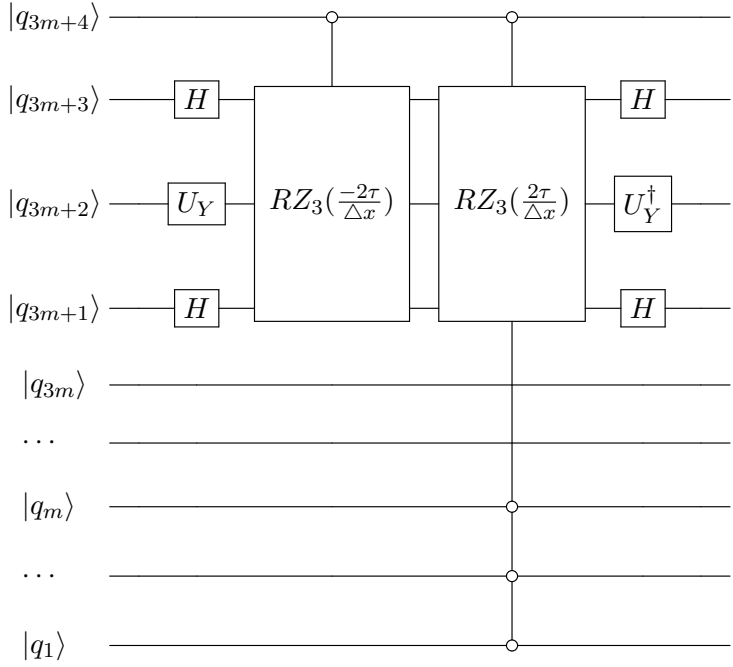
\begin{figure}
    \centerline{
     \Qcircuit @C=1em @R=2em {
					\lstick{\hbox to 3em{$\ket{q_{3m+4}}$\hss}}    &\qw   &\qw       & \ctrlo{1}  & \ctrlo{1}     &\qw   &\qw & \qw        \\
					\lstick{\hbox to 3em{$\ket{q_{3m+3}}$\hss}}    &\qw   & \gate{H}  &\multigate{2}{RZ_3(\frac{-2\tau}{\triangle x})}  &\multigate{2}{RZ_3(\frac{2\tau}{\triangle x})}   &\gate{H} &\qw   &\qw   \\
					\lstick{\hbox to 3em{$\ket{q_{3m+2}}$\hss}}     &\qw  &\gate{U_Y} & \ghost{RZ_3(\frac{-2\tau}{\triangle x})}  & \ghost{RZ_3(\frac{2\tau}{\triangle x})}   &\gate{U_Y^{\dagger}} &\qw &\qw   \\
					\lstick{\hbox to 3em{$\ket{q_{3m+1}}$\hss}}    &\qw  &\gate{H}     & \ghost{RZ_3(\frac{-2\tau}{\triangle x})} & \ghost{RZ_3(\frac{2\tau}{\triangle x})}     & \gate{H}&\qw  &\qw \\
					\lstick{\hbox to 2.5em{$\ket{q_{3m}}$\hss}}   &\qw  &\qw  & \qw  & \qw    & \qw   & \qw& \qw \\
					\lstick{\hbox to 2.5em{$\cdots$\hss}}             &\qw  &\qw & \qw  & \qw    & \qw  & \qw & \qw\\ 
					\lstick{\hbox to 2.5em{$\ket{q_{m}}$\hss}}      &\qw  &\qw &\qw  & \ctrlo{-3}          & \qw     & \qw & \qw\\
					\lstick{\hbox to 2.5em{$\cdots$\hss}}            &\qw  & \qw  & \qw & \ctrlo{-1} & \qw& \qw & \qw \\
					\lstick{\hbox to 2.5em{$\ket{q_1}$\hss}}       &\qw &\qw &\qw & \ctrlo{-1}    & \qw& \qw& \qw 
     }
     }
     \caption{Quantum circuit for $\Cu_{x_2}$.}
		\label{fig:Ux2}
	\end{figure}
	
	An argument similar to the computation of $\Cu_x$ shows the approximation of $\Cu_y$ as 
 	\begin{align}
		\Cu_y  &\approx \Cv_y = 
		\Cv_{y_1}\cdot\Cu_{y_2}, \label{eq:Vy}\\
		\Cv_{y_1} & = \tilde U_{ZYZ}\cdot \prod_{j=1}^m
        \left(U_{3m+3}^{\Ci_{y,j}}(0) X_{3m+3,j}
		\CRZ_{{3m+1\to3m+3}}^{\Ci_{y,j}^{1},\Ci_{y,j}^0}(\frac{2\tau}{\triangle y}) X_{3m+3,j} (U_{3m+3}^{\Ci_y,j}(0))^{\dagger} \right)\cdot 
		\tilde U_{ZYZ}^{\dagger}, \notag \\
		\Cu_{y_2} &= 
		\tilde U_{XYZ} \cdot\CRZ_{{3m+1\to 3m+3}}^{\emptyset, \{3m+4\}}(\frac{2\tau}{\triangle y}) \cdot \CRZ_{{3m+1\to 3m+3}}^{\emptyset, \{m+1,\cdots,2m,3m+4\}}(\frac{-2\tau}{\triangle y}) \cdot \tilde U_{XYZ}^{\dagger}, \notag
	\end{align}
	where $\Ci_{y,j} =\{m+1,m+2,\cdots,m+j\}$,  and 
	\begin{align*}
	\Ci_{y,j}^0=\{3m+4\}, \quad \Ci_{y,j}^1=\{m+1,m+2,\cdots,m+j\}.
	\end{align*}
	Here $\tilde U_{ZYZ} = \onebf^{\otimes 2} \otimes  U_Y \otimes \onebf^{\otimes (3m+1)}$,
	and
	$\tilde U_{XYZ} = \onebf\otimes H \otimes U_Y \otimes \onebf^{\otimes(3m+1)}$. 
	Similarly, one has the approximation of circuit of $\Cu_z$ shown as 
	\begin{align}
		\Cu_z  &\approx \Cv_z = \Cv_{z_1}\cdot \Cu_{z_2},\label{eq:Vz}\\
		\Cv_{z_1} & =  \tilde U_{Z\onebf Y}\cdot \prod_{j=1}^m \left( U_{3m+3}^{\Ci_{z,j}}(0) X_{3m+3,j}
		\CRZ_{{3m+1\to 3m+3}}^{\Ci_{z,j}^{1},\Ci_{z,j}^0}(\frac{-2\tau}{\triangle z}) X_{3m+3,j} (U_{3m+3}^{\Ci_{z,j}}(0))^{\dagger} \right)\cdot 
		\tilde U_{Z\onebf Y}^{\dagger}, \notag \\
		\Cu_{z_2} &= 
		\tilde U_{X\onebf Y}\cdot  \CRZ_{{3m+1\to 3m+3}}^{\emptyset, \{3m+4\}}(\frac{-2\tau}{\triangle z}) \cdot \CRZ_{{3m+1\to 3m+3}}^{\emptyset, \{2m+1,\cdots,3m,3m+4\}}(\frac{2\tau}{\triangle z}) \cdot 
		\tilde U_{ X\onebf Y}^{\dagger}, \notag
	\end{align}
	where $\Ci_{z,j} =\{2m+1,2m+2,\cdots,2m+j\}$ and 
	\begin{align*}
	\Ci_{z,j}^0=\{3m+4\}, \quad \Ci_{z,j}^1=\{2m+1,2m+2,\cdots,2m+j\}.
	\end{align*}
	Here $\tilde U_{Z\onebf Y} = \onebf^{\otimes3}\otimes U_Y \otimes \onebf^{\otimes 3m}$ and 
	$\tilde U_{X\onebf Y} = H\otimes \onebf \otimes U_Y\otimes \onebf^{\otimes 3m}$.
	
	In summary, the approximation of the time evolution $\exp(-i H\tau)$ with $H$ defined in 
	\eqref{eq:schro of tilde w} is given by
	\begin{equation}\label{eq:unitary}
		V(\tau) =\Cu_1\Cv_2\left(\Cv_x\Cv_y\Cv_z\otimes \onebf^{\otimes(n_s+n_p)}\right)= \Cu_1\Cv_2\prod_{\alpha=x,y,z}\big(\Cv_{\alpha_1}\,\Cu_{\alpha_2}\big)\otimes \onebf^{\otimes (n_s+n_p)},
	\end{equation}
	where $\Cu_1$ is defined in \eqref{eq:u1_t}-\eqref{eq:u1_2t}, $\Cv_2$ is defined in \eqref{eq:V2}, and $\Cv_x$, $\Cv_y$ and $\Cv_z$ are defined in \eqref{eq:Vx},
	\eqref{eq:Vy} and \eqref{eq:Vz}, respectively.
 
	\section{Complexity analysis}
 
	In this section, we analyze the complexity of the previously constructed quantum circuits. We begin by demonstrating the approximation of the circuit implementation for $\exp(i\Ch\tau )$ with a time step $\tau$, where 
	\begin{equation}\label{eq:def H DM}
		\Ch= \gamma \sum_{\alpha=1}^d\sum_{j=1}^n \eta_{\alpha} (e^{i\lambda_{\alpha}} (s_j^{-})_{\alpha} + e^{-i\lambda_{\alpha}} (s_j^+)_{\alpha})
	\end{equation} 
	as proved in \cite{SKHOY24}.
	Here $\eta_{\alpha}$ and $\gamma$ are  real scalar parameters, while $\lambda_{\alpha}\in \bbR$ is the phase parameter and 
	\begin{equation*}
		(s_j^{\mu})_{\alpha} = \onebf^{\otimes (\alpha-1)n} \otimes s_j^{\mu} \otimes \onebf^{\otimes (d-\alpha) n},
	\end{equation*}  
	for $\mu\in \{+,-\}$.
 
\begin{lemma}\label{lem:err of s}
		The time evolution operator  $\exp(-i \Ch \tau)$, with $\Ch$ defined in \eqref{eq:def H DM},  can be approximated by the unitary
		\begin{equation*}
			\bigotimes_{\alpha=1}^d V(\gamma \eta_{\alpha}\tau, \lambda_{\alpha}) =
			\prod_{\alpha=1}^d \onebf^{\otimes(\alpha-1)n}\otimes \left( \prod_{j=1}^n \onebf^{\otimes (n-j)}\otimes U_n^{\Ci_j}(\lambda) \CRZ_j^{\Ci_j^1,\emptyset}(2\gamma \tau) (U_n^{\Ci_j}(-\lambda))^{\dagger}
            \right) \otimes \onebf^{\otimes (d-\alpha)n},
		\end{equation*}
		where $\Ci_j = \Ci_j^1 =\{1,2,\cdots,j-1\}$, $\CRZ_j^{\Ci_j^1,\emptyset}(2\gamma \tau) = \exp(-i \gamma \tau Z\otimes \soo^{\otimes (j-1) })$ is a multi-controlled RZ($2\gamma \tau$) gate acting on the $j$-th qubit when the $1,\cdots, (j-1)$-th qubits become $1$. 
		The approximation error is upper bounded in the sense of the operator norm as  
		\begin{equation*}
			\|\exp(-i\Ch\tau ) - \bigotimes_{\alpha=1}^d V(\gamma\eta_{\alpha}\tau,\lambda_{\alpha})\| \leq \frac{\gamma^2 \tau^2 (n-1) }{2}\sum_{\alpha=1}^d \eta_{\alpha}^2.
		\end{equation*}
	\end{lemma}
 
\begin{lemma}\label{lem:Utau-Vtau}
		Consider the Schr\"odingerized  equation $\frac{\D}{\D t} \ket{\tilde\Bv_h} = -i H \ket{ \tilde \Bv_h}$ with the Hamiltonian $H$ given in  \eqref{eq:schro of tilde w}.
		The time evolution operator $\exp(-iH\tau)$ with the time step $\tau$ can be approximated by the unitary $V(\tau)$ in \eqref{eq:unitary}. The approximation error in the sense of the operator norm is upper bounded as 
		\begin{equation}
			\|\exp(-i H \tau) - V(\tau)\| \lesssim  \frac{d\tau^2 (m-1)}{\triangle x^2} 
			+ \frac{\tau^2 F_{\max}}{\triangle p \triangle x}
			+\frac{\tau^2 F_{\max}}{\triangle p \triangle s}
			+\frac{\tau^2}{\triangle s \triangle x}
            +\frac{(d+1)|\Ci||\Ci_s|F_{\max}^2\tau^2}{\triangle p},
		\end{equation}
		where $d=3$ is the dimension, $|\Ci| =\max\limits_{\alpha=x,y,z,\rho,l\in \Ci_s}|\Ci_{\alpha}^l|$, $F_{\max} =\|\BF\|_{\max} $.
	\end{lemma}
	\begin{proof}
		Using Trotter's splitting and the triangle inequality, it yields
		\begin{align}
			&\|\exp(-iH\tau) -V(\tau)\| 
			 \notag \\
			\lesssim  &\tau^2\big(\|[H_{D_s}+H_F,H_{\text{curl}}]\| 
			+\|[H_{D_s},H_F]\|
			\big)
            +\|\Cu_2 - \Cv_2\|
			+\|\Cu_3-(\Cv_x\Cv_y\Cv_z)\otimes \onebf^{\otimes (n_s+n_p)}\|.
			\label{eq:expiH-V}
		\end{align}
		In the following, we consider the approximation of the right-hand sides of  \eqref{eq:expiH-V} term by term.
		
		Under the assumption that $\triangle x = \triangle y = \triangle z$, one has
		\begin{align} 
			\|[H_{D_s}+H_F,H_{\text{curl}}]\|
			&\lesssim \|H_{D_s}+H_F\|\|H_{\text{curl}}\| \lesssim \frac{F_{\max}}{\triangle p \triangle x} + \frac{1}{\triangle s \triangle x}, \label{eq:commutor Hxy Hcurl}\\
			\|[H_{D_s},H_F]\| &\lesssim \|H_{D_s}\|\|H_F\|\lesssim \frac{F_{\max}}{\triangle s \triangle p}.
			\label{eq:commutor Hxy}
		\end{align}
        The error for $\|\Cu_2 - \Cv_2\|$ is bounded by 
        \begin{equation}
        \begin{aligned}
            \|\Cu_2 - \Cv_2\| &\lesssim \sum_{\substack{\alpha=x,y,z,\rho\\l\in \Ci_s j\in \Ci_{\alpha}^l }}\big(\|[H_{X_p,l_j}^{\alpha},H_{Y_p,l_j}^{\alpha}]\|
            +\|[H_{X_p,l_j}^{\alpha},H_{Y_p,l_j}^{\alpha}]\|
            +\|[H_{X_p,l}^{\alpha},H_{Y_p,l_j}^{\alpha}]\|
            +\|[H_{X_p,l}^{\alpha},H_{Y_p,l}^{\alpha}]\|
            \big)\\
            &\lesssim (d+1)|\Ci||\Ci_s|
            \frac{F_{\max}^2}{\triangle p},
            \end{aligned}
        \end{equation}
     where $|\Ci| = \max\limits_{\alpha=x,y,z,\rho} \max\limits_{l\in \Ci_s}|\Ci^{l}_{\alpha}|$.
		This leads from Lemma~\ref{lem:err of s} to the estimation of the last term in Equation \eqref{eq:expiH-V} as follows
        \begin{align}
			\|\Cu_3 - (\Cv_x\Cv_y\Cv_z)\otimes \onebf^{\otimes (n_s+n_p)}\|
			&\leq \|\Cu_3-(\Cu_x\Cu_y\Cu_z)\otimes \onebf^{\otimes (n_s+n_p)}\|+\|\Cu_x-\Cv_x\|+\|\Cu_y-\Cv_y\|
			+\|\Cu_z-\Cv_z\| \notag \\
			&\lesssim \tau^2(\|[H_x+H_y,H_z]\| + \|[H_x,H_y]\|)+\sum_{\alpha=x,y,z}\|\Cu_{\alpha_1}-\Cv_{\alpha_1}\| \notag \\
			&\lesssim \frac{\tau^2}{\triangle x^2} + \frac{d\tau^2 (m-1)}{\triangle x^2}. \label{eq:Cu3}
		\end{align}
		Inserting  \eqref{eq:commutor Hxy Hcurl}--\eqref{eq:Cu3} into \eqref{eq:expiH-V}, the proof is completed.
	\end{proof}
 
	\begin{lemma}\label{lem:U-V}
		The approximation of $\exp(-iH\tau)$ denoted by 
		$V(\tau)$, where $H$ is the Hamiltonian defined 
		in \eqref{eq:schro of tilde w} and $\tau$ is the time step, can be implemented by $\mathscr{O}\left(dm+n_s+(d+1)|\Ci||\Ci_{s}|n_p\right)$ single-qubit gates and 
		at most $\mathscr{O}\left(dm^2+n_s^2+(d+1)|\Ci||\Ci_{s}|n_p(n_s+md)\right)$ CNOT gates,
		where $|\Ci| = \max\limits_{\alpha=x,y,z,\rho} \max\limits_{l\in \Ci_s}|\Ci^{l}_{\alpha}|$ and 
		$n_p,n_s,m\geq d=3$.
	\end{lemma}
 
	\begin{proof}
		Let $\Cn_{\text{S}}(\Cu)$ be the number of single-qubit gates of the operator $\Cu$ and  
		$\Cn_{\text{CNOT}}(\Cu)$ be the number of non-local gates of the operator $\Cu$. It is obvious that
		\begin{equation*}
			\Cn_{\beta}(V(\tau)) = \Cn_{\beta}(\Cu_{1})+\Cn_{\beta}(\Cv_{2})+\sum_{\alpha=x,y,z}\big( \Cn_{\beta}(\Cv_{\alpha_1})+
			\Cn_{\beta}(\Cv_{\alpha_2})\big),
		\end{equation*}
		where $\beta=\{\text{S},\text{CNOT}\}$.
		According to Figure~\ref{tilde U1}, the implementation of $\Cu_1$ just has a maximum of 
		\begin{align}\label{eq:Cu1}
			\Cn_{\text{S}}(\Cu_1)  = \mathscr{O}(n_s+1) 
		\end{align}
		single gates. The non-local gates included in the operator $\Cu_1$ are $\mathscr{O}(n_s^2)$
        CR gates to implement the (inverse) quantum  Fourier transform (QFT, IQFT), which corresponds to  $\mathscr{O}(n_s^2)$ CNOT gates.

        The number of single qubit gates and CNOT gates of $\Cv_2$ satisfies 
        \begin{equation}\label{eq:N_u2_alpha}
            \Cn_{\beta}(\Cv_2) = \sum_{l\in \Ci_s}\sum_{j\in \Ci_{\alpha}^l}\left(
            \Cn_{\beta}(\Cu_{X_p,l_j}^{\alpha})
            +\Cn_{\beta}(\Cu_{Y_p,l_j}^{\alpha})\right)
            +\left(\Cn_{\beta}(\Cu_{X_p,l}^{\alpha})+\Cn_{\beta}(\Cu_{Y_p,l}^{\alpha})\right),
        \end{equation}
        where $\beta\in \{\text{S},\text{CNOT}\}$.
		From Figures  \ref{fig:tilde UX_lj}--\ref{tilde U2l_gamma12}, one can find that the operator  $\Cu_{X_p,l_j}^{\alpha}$ has $n_p$   $\Cu_{X,l_j}^{\alpha}$ which is a multi-controlled RX gate,  each consisting of $n_s+m(d+1)$ control points, and 
        $\Cu_{X_p,l}^{\alpha}$ has $n_p$ multi-controlled RX gates, each of which has $n_s+3$ control points.
        It is known from \cite{Vale24} that the multi-controlled RZ gate with $(j-1)$ control qubits can be decomposed into single-qubit gates and at most $16j-40$ CNOT gates. Therefore, the number of single-qubit and CNOT gates required to implement the operator $\Cu_{X_p,l_j}^{\alpha}$ and $\Cu_{X_p,l}^{\alpha}$ are 
        \begin{align}
            \Cn_{\text{S}}(\Cu_{X_p,l_j}^{\alpha})&=n_p,\quad 
            \Cn_{\text{CNOT}}(\Cu_{X_p,l_j}^{\alpha}) =  n_p(16(n_s+m(d+1)+1)-24),\label{eq:N_Uxp}\\
             \Cn_{\text{S}}(\Cu_{X_p,l}^{\alpha}) &= n_p, \quad 
             \Cn_{\text{CNOT}}(\Cu_{X_p,l}^{\alpha}) = n_p(16(n_s+d+1)-24).
        \end{align}
        From Equation \eqref{eq:UY_p}, one gets
        \begin{align}
            \Cn_{\text{S}}(\Cu_{Y_p,l_j}^{\alpha}) &= 1, \quad 
            \Cn_{\text{CNOT}}(\Cu_{Y_p,l_j}^{\alpha}) = 16(n_s+m(d+1))-24,\\
            \Cn_{\text{S}}(\Cu_{Y_p,l}^{\alpha}) &= 1, \quad 
            \Cn_{\text{CNOT}}(\Cu_{Y_p,l}^{\alpha}) = 16(n_s+d+1)-24. \label{eq:N_Uyp}
        \end{align}
        Inserting Equation \eqref{eq:N_Uxp}--\eqref{eq:N_Uyp} into Equation \eqref{eq:N_u2_alpha}, the number of gates to implement $\Cv_2$ are
        \begin{align}
            \Cn_{\text{S}}(\Cv_2) = \mathscr{O}\left((d+1)|\Ci_s||\Ci|(n_p+1)\right),\quad 
            \Cn_{\text{CNOT}}(\Cv_2) = \mathscr{O}\left((d+1)|\Ci_s||\Ci|(n_pn_s+dn_p m)\right),
        \end{align}
        where $|\Ci|=\max\limits_{\alpha=x,y,z,\rho,l\in \Ci_s}|\Ci_{\alpha}^l|$, $d=3$.

		As shown in Figure~\ref{fig:cirUx1}, the operator $W_j$ consists of  a multi-controlled $\text{RZ}_3$ gate with  $j+1$ control points
		and a total of $2(j-1)$ CNOT gates. 
		Therefore, the number of CNOT gates and single-qubit gates required to implement the approximated $\Cv_{x_1}$
		is 
		\begin{equation}\label{eq:NVx1}
			\Cn_{\text{S}} (\Cv_{x_1})= \mathscr{O}(m),\quad 
   \Cn_{\text{CNOT}}(\Cv_{x_1}) =\mathscr{O}(m^2). \quad 
		\end{equation}
		From Figure \ref{fig:Ux2}, it yields similarly 
		\begin{equation}\label{eq:NUx2}
			\Cn_{\text{S}}(\Cu_{x_2}) = \mathscr{O}(1),\quad 
   \Cn_{\text{CNOT}}(\Cu_{x_2}) = \mathscr{O}(m+1). 
		\end{equation}
		In summary, from \eqref{eq:Cu1}-\eqref{eq:NUx2}, one has 
		\begin{align*}
			\Cn_{S}(V(\tau)) &= \mathscr{O}\left(dm+n_s+(d+1)|\Ci||\Ci_{s}|n_p\right), \\ 
			\Cn_{\text{CNOT}}(V(\tau))& = \mathscr{O}\left(dm^2+n_s^2+(d+1)|\Ci||\Ci_{s}|n_p(n_s+md)\right),
		\end{align*}
  where $d=3$. The proof is finished.
	\end{proof}

 \subsection{Main results}
 
 In this subsection, we present the main result of the complexity of the Schr\"odingerisation for Maxwell's equations. Before that, we give the 
 error estimates of the spectral methods (see for example \cite{Shen11}). 
 
 \begin{lemma}\label{them:error of Pn and In}
 	For $u\in H_p^m(I)$ which consists of 
 	functions with derivatives of order up to $m-1$ being $2\pi L$-periodic,
 	there holds
 	\begin{equation}\label{eq:hat ek}
 		\|u-\Ci_N u\|_{L^2(I)} \lesssim \triangle p^{m}|u|_{H^{m}(I)},
 	\end{equation}
 	where $\Ci_N$ is the discrete Fourier interpolation by 
 	\begin{equation}
 		\Ci_N u(p) = \sum_{k=-N/2}^{N/2} \tilde u_k e^{ik(p/L+\pi)}, \quad 	\tilde u_k = \frac{1}{N c_k} \sum_{j=0}^{N-1}u(p_j)e^{-ik(p_j/L+\pi)},\quad k=-N/2,\cdots,N/2,
 	\end{equation}
 	where $c_k =1$ for $|k|<N/2$, and $c_k = 2$ for $k=\pm N/2$.
 \end{lemma}
 
 \begin{lemma}\label{lem:uhd-uh}
 	Assume $L$ and $S$ are large enough such that  $e^{-S\pi +T} \approx e^{-L\pi +T/2} \approx 0$, and
    $g(p)\in C^r(\bbR)$, $\beta(\xi)\in H^r(\bbR)$
    is the $r$-th order approximation to the $\delta$ function. 
      Define $\Bv_h^D(T) = \Cv \Bv_h(0)$ and $\Bu_h^D \ket{k} = e^{p_k}\triangle s\sum\limits_{l\in [N_s]} (I_{n} \otimes \bra{l}\otimes I_{N_p} )M_k \Bv_h^D$, where $\Cv = (I_n\otimes I_{N_s}\otimes \Phi^p) V^{N_t}(\tau)(I_n\otimes I_{N_s}\otimes (\Phi^p)^{\dagger} )$.
 	The relative error holds as follows
 	\begin{equation}
 			\frac{\|\Bu_h^D\ket{k} - \Bu_f \ket{k}\|}{\|\Bu_f(T)\|}  \lesssim \frac{e^{p_k}\|\Bu_f(0)\|}{\|\Bu_f(T)\|} \|\Cv-\Cu\|
 			+\triangle p^r + \triangle s^r,
 		\end{equation}
 		where $\Cu = (I_n \otimes I_{N_s} \otimes \Phi^p ) U  (I_n \otimes I_{N_s} \otimes (\Phi^p)^{\dagger} )$, and $p_k\geq T/2$.
 \end{lemma}
 
 \begin{proof}
Since $e^{-\pi L+T/2} \approx 0 $ and $g(p)\in C^r(\bbR)$ , the error bounded by the spectral method is obtained from Lemma \ref{them:error of Pn and In} such that 
 	\begin{equation}\label{eq:err of wh -w}
 		\frac{\|\mathbf{M}_k\Bw_h - e^{-p_k}\Bu_f\ket{k}\|}{\|\Bw(t,p_k)\|} \leq  \mathscr{O}(\triangle p^r),
 	\end{equation}
 	where  $\mathbf{M}_k = I_n \otimes \ket{k}\bra{k}$ and we use 
 	$\Bw(t,p_k) = e^{-p_k}\Bu$ for $p_k \geq T/2$ from the recovery Theorem~\ref{thm:recovery u}.
 	
 Since $\Bv$  satisfies the  transport equation in the $s$ direction, and $\beta(\xi)\in H^r(\bbR)$  is the rth-order approximation to the $\delta$ function, one gets the error between the semi-discrete solution  $\Bw_h^D$ in Equation \eqref{eq:whd} and $ \Bw_h$ in 
 Equation \eqref{eq:v(s)}-\eqref{eq:tilde whd} from Lemma~\ref{them:error of Pn and In}
 \begin{equation}\label{eq:err of whd -wh}
 	\frac{\|\mathbf{M}_k(\Bw_h^D -\Bw_h)\|}{\|\Bw(t,p_k)\|} \leq \mathscr{O}(\triangle s^r ),
 \end{equation}
 under the assumption $e^{-\pi S+ T}\approx 0$. Here we used the relation $\|\mathbf{M}_k\Bw_h\| =  (1+\mathscr{O}(\triangle p^r))\| \Bw(t,p_k)\|$ from Equation \eqref{eq:err of wh -w}.
 According to Equation\eqref{eq:err of wh -w}-\eqref{eq:err of whd -wh}, one has the relative error as follows
\begin{equation*}
	\begin{aligned}
		&\frac{\|\Bu_h^D\ket{k}-\Bu_f\ket{k}\|}{\|\Bu_f(T)\|} = \frac{\|e^{p_k}\triangle s\sum\limits_{l\in [N_s]} (I_{n} \otimes \bra{l}\otimes I_{N_p} )M_k \Cv \Bv_h(0) - \Bu_f \ket{k}\|}{\|\Bu_f(T)\|}\\
		\vspace{1mm}
		\leq & \frac{e^{p_k}}{\|\Bu_f(T)\|}\big\|
		 \triangle s\sum\limits_{l\in [N_s]} (I_{n} \otimes \bra{l}\otimes I_{N_p} )M_k(\Cv-\Cu)\Bv_h(0)\big\|\\
		& + \frac{1}{\|\Bw(T,p_k)\|}\big \| 
		 \triangle s\sum\limits_{l\in [N_s]} (I_{n} \otimes \bra{l}\otimes I_{N_p} )M_k\Cu \Bv_h(0)
		- \mathbf{M}_k\Bw_h
		\big\| 
		+ \big\|
		\frac{e^{p_k}}{\|\Bu_f(T)\|} \Bw_h(t,p_k) -\ket{\Bu_f} \ket{k}
		\big\| \\ 
		\lesssim & \frac{e^{p_k} \|\Bu_f(0)\|}{\|\Bu_f(T)\|}\|\Cv - \Cu\| + \triangle p^r + \triangle s^r.
	\end{aligned}
\end{equation*}
Here we have used  $\Bw(t,p_k) = e^{-p_k} \Bu_f(t)$ and $\|\triangle s \sum_{l\in [N_s]} (I_{n} \otimes \bra{l}\otimes I_{N_p}) M_k \Bv_h(0)\| = \mathscr{O}(\|\Bu_f(0)\|)$.
\end{proof}
    
	\begin{theorem}\label{thm:computation}
	Given Maxwell's equations with time dependent source term,  then $\Bu_f(t)$, the solution of Equation \eqref{eq:ODE1} with a mesh size $\triangle x = \triangle y = \triangle z = 1/M$, $m = \log_2 M$,
	can be prepared with  precision $\varepsilon$ using the Schr\"odingerisation method depicted in Figure~\ref{schr_circuit}. 
	Assume the assumptions in Lemma~\ref{lem:uhd-uh} and Assumption~\ref{ass:J} hold,
	this preparation can be achieved using at most 
	$$\tilde{\mathscr{O}}\left( 
				 \frac{T^2e^{\frac{3}{2}} \|\Bu_f(0)\|^3}{\varepsilon\|\Bu_f(T)\|^3}
            \big( M^2d^2m(|\Ci||\Ci_s|+m)+\frac{(dF_{\max}|\Ci||\Ci_s|)^2}{\sqrt[r]{\varepsilon}}\big)
	\right)$$ single-qubit gates and
	$$\tilde{\mathscr{O}}\left(
	       \frac{T^2e^{\frac{3}{2}} \|\Bu_f(0)\|^3}{\varepsilon\|\Bu_f(T)\|^3} 
        \big( M^2d^2m^2(d|\Ci||\Ci_s|+m)+\frac{md(dF_{\max}|\Ci||\Ci_s|)^2}{\sqrt[r]{\varepsilon}}\big)
	\right)$$ CNOT gates.
\end{theorem}

		\begin{proof}
			We divide the time into equal intervals, with the time step defined as $\tau = T/N_t$. 
			For the simulation accuracy to be within $\varepsilon$,
			according to Lemma~\ref{lem:uhd-uh}, 
            one needs to choose the number of simulation steps $N_t$  and mesh size $\triangle p$,
            $\triangle s$ to satisfy the following conditions
			\begin{align} 
			&	\|\Cu - \Cv\| =\|U^{N_t}(\tau) - V^{N_t}(\tau)\| \leq N_t\|U(\tau) - V(\tau)\|\leq  \frac{e^{-p_k}\|\Bu_f(T)\|}{\|\Bu_f(0)\|}\varepsilon, \label{eq:Cu-Cv}\\
			&	\triangle p \sim \triangle s \sim \sqrt[r]{\varepsilon}. \label{eq:triangle p, s}
			\end{align}
		    Combined with Lemma~\ref{lem:Utau-Vtau}, $N_t$ should be large enough such that
			\begin{equation*}
			N_t \geq \frac{T^2 e^{p_k} \|\Bu_f(0)\|}{\varepsilon \|\Bu_f(T)\|} \bigg(
			\frac{d(m-1)}{\triangle x^2} + \frac{F_{\max}}{\sqrt[r]{\varepsilon}\triangle x} + \frac{F_{\max}}{\sqrt[r]{\varepsilon^2}}
			+\frac{1}{\sqrt[r]{\varepsilon} \triangle x} +\frac{F_{\max}^2 |\Ci||\Ci_s| d}{\sqrt[r]{\varepsilon}}
			\bigg).			
			\end{equation*}
		    From Equation \eqref{eq:triangle p, s}, one has 
		\begin{align}
			N_p  &= \frac{2\pi L}{\triangle p} = \mathscr{O}(\frac{\ln(\frac{1}{\varepsilon})}{\sqrt[r]{\varepsilon}}), \quad n_p = \log_2 N_p = \mathscr{O} (\log (\frac{\ln(\frac{1}{\varepsilon})}{\sqrt[r]{\varepsilon}}) ),\\
			N_s  &= \frac{2\pi S}{\triangle s} = \mathscr{O}(\frac{\ln(\frac{1}{\varepsilon})
            }{\sqrt[r]{\varepsilon}}),\quad n_s = \log_2 N_s = \mathscr{O}(\log (\frac{\ln(\frac{1}{\varepsilon})}{\sqrt[r]{\varepsilon}})),
		\end{align}
	where we have used the assumption $e^{-\pi S+T} \approx e^{-\pi L + \frac{T}{2}} \approx 0 \leq \varepsilon$.
	Therefore, the numbers of single-qubit gates and CNOT gates to implement $V^{N_t}(\tau)$ are, respectively,  
			\begin{align*}
				&\Cn_{\text{S}} ( V^{N_t}(\tau)) 
				= \mathscr{O} (N_t \Cn_{\text{S}}(V(\tau)))\\
				 \leq &\mathscr{O}\bigg(\frac{T^2e^{p_k} \|\Bu_f(0)\|}{\varepsilon\|\Bu_f(T)\|}
                  \big( d|\Ci| |\Ci_s|\log(\frac{\ln(\frac{1}{\varepsilon})}{\sqrt[r]{\varepsilon}}) +md\big) 
				 \cdot
				  \big(M^2 dm + \frac{F_{\max}(M+dF_{\max}|\Ci_s||\Ci|)}{\sqrt[r]{\varepsilon}}
				 \big)\bigg)\\
				 = &\tilde{\mathscr{O}}\bigg( 
				 \frac{T^2e^{p_k} \|\Bu_f(0)\|}{\varepsilon\|\Bu_f(T)\|}
                   \big( M^2d^2m(|\Ci||\Ci_s|+m)+\frac{(dF_{\max}|\Ci||\Ci_s|)^2}{\sqrt[r]{\varepsilon}}\big)
                   \bigg),\\
				&\Cn_{\text{CNOT}} (V^{N_t}(\tau)) =
				 \mathscr{O}(N_t \Cn_{\text{CNOT}}(V(\tau)))  \\
			    \leq &\tilde{\mathscr{O}}\bigg(
	           \frac{T^2e^{p_k} \|\Bu_f(0)\|}{\varepsilon \|\Bu_f(T)\|} 
              \big( M^2d^2m^2(d|\Ci||\Ci_s|+m)+\frac{md(dF_{\max}|\Ci||\Ci_s|)^2}{\sqrt[r]{\varepsilon}}\big)
	           \bigg).
 			\end{align*}
 			Next, we consider the probability of getting the desired state 
 			\begin{equation*}
 				M_k \ket{\Bv_h^D} \approx \frac{e^{-p_k}\Bu_f(T)\ket{k}\otimes \deltabf_h }{\|\Bv_h^D(T)\|} \approx \frac{e^{-p_k} \|\Bu_f(T)\|}{\|\Bu_f(0)\|} \ket{\Bu_f(T)}\ket{k}\ket{\deltabf_h(T)}.
 			\end{equation*}
 			The first approximation is from Lemma~\ref{lem:uhd-uh} and Equation \eqref{eq:Cu-Cv}.
 			Therefore, $\mathscr{O}(\frac{e^{2p_k}\|\Bu_f(0)\|^2}{\|\Bu_f(T)\|^2})$ measurements are needed.
 			The proof is finished by choosing $p_k = T/2$.
		\end{proof}
  
 \begin{remark}\label{eq:probalility of uf}
	From \eqref{eq:matrix FDTD}, by applying  Duhamel's principle, 
	one gets 
	\begin{equation*}
		\Bu_T = \Bu(T) = e^{AT}\Bu_0 + \int_0^T e^{A(T-s)}\Bf(s)\;ds.
	\end{equation*}
	Since $A$ is skew Hermitian, one has 
    \begin{equation*}
        \|\Bu_T\|^2 = \|\Bu_0\|^2 + 2 \Bu_0^{\top}\tilde \Bf + \|\tilde \Bf\|^2,
    \end{equation*}
    where $\tilde \Bf = \int_0^T e^{-As}\Bf(s)ds$ and $\|\tilde \Bf\|\leq \|\int_0^T\Bf(s)ds\|$.
    Consequently, it follows that  $\|\Bu_T\|/\|\Bu_0\| = \mathscr{O}(1)$ for any fixed $T=\mathscr{O}(1)$,
     provided that the source terms satisfy  $\BJ,\rho\in L^{\infty}(0,T;L^2(\Omega))$.
	Considering $\Bu_f(T) = \Bu(T) + \Br$, it follows 
	\begin{equation*}
		\begin{aligned}
		\frac{\|\Bu_f(0)\|^2}{\|\Bu_f(T)\|^2}& 
       = \frac{\triangle x^d\|\Bu_0\|^2+8c_0^2}{\triangle x^d \|\Bu_T\|^2+8c_0^2}
	    \approx \frac{\|\BE_0\|^2_{L^2(\Omega)} + \|\BB_0\|_{L^2(\Omega)}^2 + 8c_0^2}{\mathscr{O}(\|\BE_0\|^2_{L^2(\Omega)} + \|\BB_0\|_{L^2(\Omega)}) + 8c_0^2}	
		\end{aligned}
	\end{equation*}
    Here we used the numerical integration $\triangle x^d 
    \|\Bu_0\|^2 \approx \|\BE(0)\|_{L^2(\Omega)}^2+\|\BB(0)\|_{L^2(\Omega)}^2$. Noting that $c_0 = \max\{\max\limits_{t\in (0,T)}\|\Bf\|_{l^{\infty}},1\} \leq \max\{\max\limits_{t\in (0,T)} \BJ, \max\limits_{t\in (0,T)} \rho, 1\} = \mathscr{O}(1)$, one has $\|\Bu_f(0)\|/\|\Bu_f(T)\| = \mathscr{O}(1)$.
	Therefore, 
    introducing an auxiliary variable $\Br$ to obtain a homogeneous equation~\eqref{eq:inhomo linear Eq} does not impose a significant additional computational burden.
\end{remark}

   \subsubsection{Comparison with the classical algorithm (FDTD)}
   To achieve an approximate solution with accuracy  $\varepsilon$ using a second-order spatial discretization, 
	we choose $M = 1/\sqrt{\varepsilon}$ and  $m\sim
	\log(1/\sqrt{\varepsilon})$.
	According to Theorem~\ref{thm:computation}, the computation cost for single-qubit and CNOT gates is at most
	\begin{equation*}
		\tilde{\mathscr{O}}\left( 
		\frac{T^2e^{3T/2} \|\Bu_f(0)\|^3 d^3}{\|\Bu_f(T)\|^3} 
		\left(\frac{1}{\varepsilon^2}\left(\log\left(\frac{1}{\sqrt{\varepsilon.}}\right)\right) + \frac{F_{\max}^2}{\varepsilon^{1/r+1}}\right)
		\right) = \tilde{\mathscr{O}}\left(d^3\varepsilon^{-2}\log(\frac{1}{\sqrt[r]{\varepsilon}})\right),\quad r\geq 2.
	\end{equation*}
	In comparison,  classical algorithms, such as the explicit Finite-Difference Time-Domain (FDTD) method, incur a computational cost of $\mathscr{O}(\varepsilon^{-1-\frac{d}{2}})$. Our quantum algorithms, derived via Schr\"odingerization, demonstrate a modest advantage in computational complexity when $d=3$.

    To further enhance the advantages of quantum computing, we can employ a higher-order Trotter-Suzuki formula \cite{CMNR18}
    for $\Ch = \sum_{l=1}^k \gamma_l \Ch_l$ such that 
    \begin{equation}
        \left(\prod_{j=1}^k \exp(-i \frac{ \gamma_j \Ch_j T}{2N_t}) \prod_{j=k}^1 \exp(-i \frac{\gamma_j   \Ch_jT}{2N_t})\right)^{N_t} = \exp(-i \Ch T) + \mathscr{O}(\frac{(k h T)^3}{N_t^2}),
    \end{equation}
    where $h=\max_j \|\Ch_j\|$. 
    For the Hamiltonian   $H$
    defined in \eqref{eq:H}--\eqref{eq:Hcurl}, with $H_F$ and $H_{\text{curl}}$ reformulated in Equation~\ref{eq:H_F} and Equation \eqref{eq:Hcurl_HxHyHz}, respectively, 
    achieving an approximation with precision $\varepsilon$ requires the time step $N_t$ to satisfy
    \begin{equation}
        N_t  =\mathscr{O}\left(\frac{1}{\sqrt{\varepsilon}}(\frac{md+n_s+d|\Ci||\Ci_s|n_p}{\triangle x})^{\frac{3}{2}}\right).
    \end{equation}
    Since the number of gates required to implement the circuit is doubled, the computational cost for single-qubit and CNOT gates remains at most 
    \begin{equation}
        \tilde{\mathscr{O}}(\varepsilon^{-\frac{5}{4}} \log(\frac{1}{\sqrt[r]{\varepsilon}})),
    \end{equation}
    when $d=\mathscr{O}(1)$. This represents a significant acceleration compared to classical algorithms. 
\section{Numerical tests}
For the numerical experiments, we utilize a classical computer to simulate the Hamiltonian system, verifying the practicality of the algorithms above, particularly focusing on the high-order accuracy of the scheme and recovery from Schr\"odingerisation. To simplify the presentation,  we restrict ourselves to a reduced version of the Maxwell equations with one spatial variable, $x$,  namely 
\begin{equation*}
\begin{aligned}
    &\partial_t E_y + \partial_x B_z = -J_y,\quad 
    \partial_t B_z + \partial_x E_y = 0,\quad \text{in}\;[0,2],\\
    &\qquad \quad \quad  E_y(0) = 0,\qquad E_y(2) = 0.
    \end{aligned}
\end{equation*}
	In order to test the accuracy of the algorithm, we set the exact solution as
	\begin{equation}
        E_y = \sin(\pi(x+t))-\sin(\pi t),\quad
		B_z = -\sin(\pi(x+t)).
	\end{equation}

In this numerical test, we truncate the $p$ region to $(-4\pi,4\pi)$ and the $s$ region to $(-5,5)$. The simulation stops at $T=1/2$. 
We examine the convergence rates of Schr\"odingerization concerning $p$ and $s$ variables with fixed $\triangle x = 1/2^4$.
The result of the above simulation  is shown in Figure~\ref{fig:dis_EB}, which shows that the numerical solutions from Schr$\ddot{\text{o}}$dingerisation  are in agreement with the exact solutions.

As shown in Table~\ref{tab:err of recovery conti}, a second-order convergence is achieved owing to the smoothness of 
$g(p)\in C^2(\bbR)$, $\beta(\xi)\in H^2(\bbR)$, and the utilization of second-order temporal discretization schemes.
  \begin{table}[htbp]
	        \centering
	      \begin{tabular}{lcccccc}
		\midrule [2pt]
		$(\triangle p,\triangle s,\triangle t)$          & $(\frac{8\pi}{2^5},\frac{10}{2^5}, \frac{1}{2^6})$ & order & $(\frac{8\pi}{2^6},\frac{10}{2^6}, \frac{1}{2^7})$ & order & $(\frac{8\pi}{2^7},\frac{10}{2^7}, \frac{1}{2^8})$ & order\\ 
		\hline
		$\|E_{\text{schr}}-E_h\|_{l^{\infty}}$       
        &4.5819e-01  &-  &1.0865e-01  &2.07 &1.5440e-02 &2.81\\ \hline
		$\|B_{\text{schr}}-B_h\|_{l^{\infty}}$ 
        &4.2349e-01  &- &1.0732e-01  &1.98 &9.7667e-03 &3.45\\
		\midrule [2pt]
	\end{tabular}            
	\caption{
		The convergence rates of  $\|E_{\text{schr}}-E_h\|_{l^{\infty}}$ and $\|B_{\text{schr}}-B_h\|_{l^{\infty}}$, respectively,  where
       $E_h$ and $B_h$ are the solutions to Equation \eqref{eq:semi-Eh}--\eqref{eq:semi-Bh}, 
       $E_{\text{schr}}$ and $B_{\text{schr}}$ are the recovery from the Schr\"odingerisation according to Equation \eqref{eq:recovery uh} with 
       $k:=\min_j\{j:p_j>0.5\}$.
	}\label{tab:err of recovery conti}
\end{table}
  \begin{figure}[t]
	\centering
	\subfigure[$Ey$]{
		\includegraphics[width=0.35\linewidth]{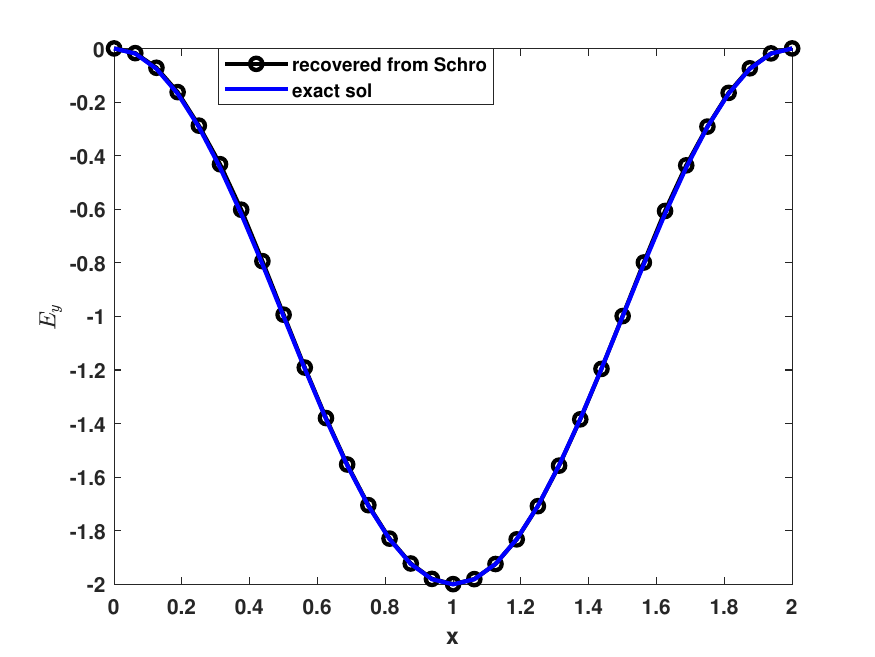}\label{fig:Ey}}
	\subfigure[$B_z$]{
		\includegraphics[width=0.35\linewidth]{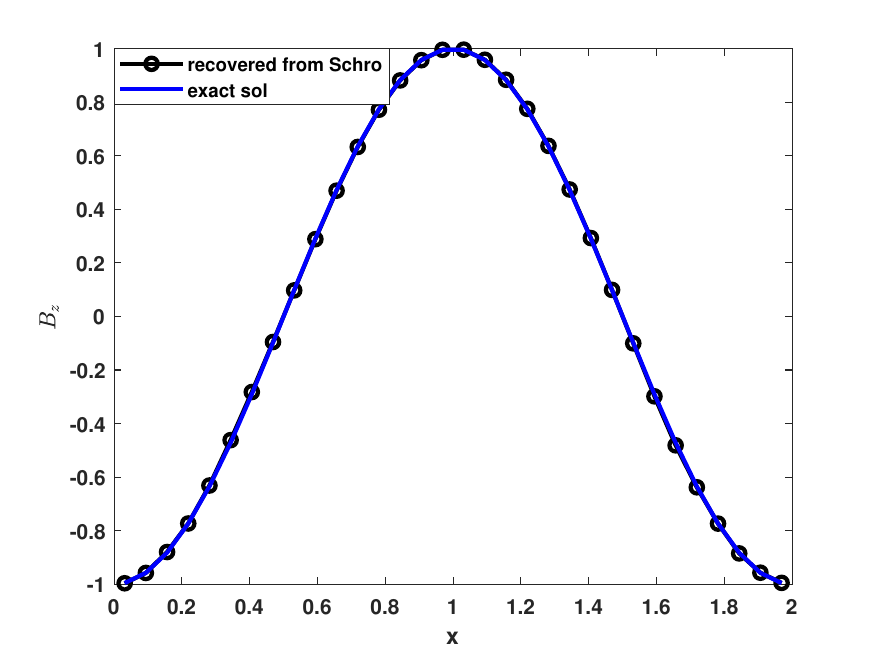}\label{fig:Bz}}
	\caption{Simulation of electromagnetic fields by Schr\"odingerisation at $T=1/2$ with the truncated domain restriced to $p\in (-4\pi,4\pi)$, $s\in (-5,5)$, and mesh size $\triangle x = \frac{1}{2^4}$, $\triangle s = \frac{10}{2^{7}}$, $\triangle p = \frac{8\pi}{2^7}$.}
	\label{fig:dis_EB}
\end{figure} 
\section{Conclusions}

The Schr\"odingerisation method, in conjunction with the autonomousization approach in \cite{cjL23}, transforms time-varying linear partial and ordinary differential equations with non-unitary dynamics into time-independent Schr\"odinger-type equations. This transformation is achieved through a warped phase transformation that maps the equations into two higher dimensions.
In this paper, we present a detailed implementation and accordingly the detailed computational complexity analysis of quantum circuits for the Schr\"odingerisation of Maxwell's equations under  Perfect Electric Conductor (PEC) boundary conditions and time-varying source terms. Through the application of smooth extension in the initial data of the extended space, and high-order approximations to the delta function, the rise in dimensionality due to the transformation does not increase the computational load of quantum computations, only $\mathscr{O}(\log \log (1/\varepsilon))$ with 
$\varepsilon$ the desired precision. 
In addition, transforming source-driven ODE systems into homogeneous systems by introducing auxiliary variables with the stretch transformation \cite{JLM24} does not diminish the success probability of obtaining the target states. 

\section*{Acknowledgements}
SJ, NL and LZ are supported by NSFC grant No. 12341104, the Shanghai Jiao Tong University 2030 Initiative and the Fundamental Research Funds for the Central Universities. SJ was also partially supported by the Shanghai Municipal Science and Technology Major Project (2021SHZDZX0102). NL also acknowledges funding from the Science and Technology Program of Shanghai, China (21JC1402900). LZ is also partially supported by the Shanghai Municipal Science and Technology Project 22JC1401600. CM was partially supported by China Postdoctoral Science Foundation (No. 2023M732248) and Postdoctoral Innovative Talents
Support Program (No. BX20230219).

\end{document}